\documentclass[11pt,reqno]{amsart}

\usepackage{amsmath, amsthm, amsfonts}
\usepackage{epsfig}
\usepackage{graphicx}

\numberwithin{equation}{section}

\newcommand{\C}{{\mathbb C}}
\renewcommand{\Re}{{\operatorname{Re\,}}}
\renewcommand{\Im}{{\operatorname{Im\,}}}
\newcommand{\al}{\alpha}
\newcommand{\be}{\beta}
\newcommand{\ga}{\gamma}
\newcommand{\Ga}{\Gamma}

\newcommand{\ep}{\varepsilon}
\newcommand{\de}{\delta}

\newcommand{\sg}{\sigma}

\newcommand{\om}{\omega}
\newcommand{\Om}{\Omega}
\newcommand{\z}{\zeta}

\newtheorem{lemma}{Lemma}[section]
\newtheorem{theorem}[lemma]{Theorem}

\newtheorem{proposition}[lemma]{Proposition}
\newtheorem{corollary}[lemma]{Corollary}
\newtheorem{remark}[lemma]{Remark}

\begin{document}

\title[Painlev\'e I double scaling limit  in the cubic random matrix model]
{Painlev\'e I double scaling limit in the cubic random matrix model}

\author{Pavel Bleher}
\address{Department of Mathematical Sciences,
Indiana University-Purdue University Indianapolis,
402 n. Blackford St., Indianapolis, In 46202, U.S.A.}
\email{bleher@math.iupui.edu}

\author{Alfredo Dea\~no}
\address{School of Mathematics, Statistics and Actuarial Science, University of Kent,
Canterbury CT2 7NF, 
United Kingdom}
\email{A.Deano-Cabrera@kent.ac.uk}

\thanks{The first author is supported in part
by the National Science Foundation (NSF) Grants DMS-0969254 and DMS-1265172. The second author acknowledges financial support from projects MTM2009--11686, from the Spanish Ministry of Science and Innovation, and from projects MTM2012--34787 and MTM2012-36732--C03--01, from the Spanish Ministry of Economy and Competitivity, and he is also grateful to the Department of Mathematical Sciences, Indiana University Purdue University Indianapolis, for their hospitality in April--May 2012 and in August 2013, when a substantial part of this work was done. The research leading to these results has received financial support from the Fund for Scientific Research Flanders through
Research Project G.0617.10 and from the European Community's Seventh Framework Programme 
(FP7/2007-2013) under grant agreement 247623--FP7-PEOPLE-2009-IRSES}

\date{\today}

\begin{abstract}
We obtain the double scaling asymptotic behavior of the recurrence coefficients and the partition function at the critical point of the $N\times N$ Hermitian random matrix model with cubic potential. We prove that the recurrence coefficients admit an asymptotic expansion in 
powers of $N^{-2/5}$, and in the leading order the asymptotic behavior of the recurrence coefficients
 is given by a Boutroux tronqu\'ee solution to
the Painlev\'e I equation. We also obtain the double scaling limit of the partition function, and we prove that the 
poles of the tronqu\'ee solution are limits of zeros of the partition function.  The tools used include
 the Riemann--Hilbert approach and the Deift--Zhou nonlinear steepest descent method for the corresponding family of complex orthogonal polynomials and their recurrence coefficients, together with  the Toda equation in the parameter space.
\end{abstract}


\maketitle
\tableofcontents 
\section{Introduction}
We consider a unitary invariant ensemble of random matrices given by the  probability distribution,
\begin{equation}\label{dPM_general}
d\mu_N(M)=\frac{1}{\tilde Z_N}e^{-N \textrm{Tr} V(M)}dM,
\end{equation}
where $V(M)$ is a polynomial,
\begin{equation}
V(M)=\sum_{k=1}^{2\nu} t_k M^k, \qquad t_{2\nu}>0.
\end{equation}
Here $M$ belongs to the space $\mathcal H_N$ of $N\times N$ Hermitian matrices
 and $dM$ is the Lebesgue measure on the space $\mathcal H_N$. The partition function of this model is equal to
\begin{equation}\label{tZn_general}
\tilde Z_N=\int_{\mathbb{R}}\ldots\int_{\mathbb{R}} e^{-N \textrm{Tr} V(M)}dM.
\end{equation}

The probability distribution of eigenvalues is
\begin{equation}
d\mu_N(z_1,\ldots,z_N)=\frac{1}{Z_N} \prod_{1\leq j<k\leq N}(z_j-z_k)^2\,
\prod_{j=1}^N e^{-N V(z_j)}dz_1\ldots dz_N,
\end{equation}
with the partition function
\begin{equation}\label{ZN0}
Z_N=\int_{\mathbb{R}}\ldots\int_{\mathbb{R}} \prod_{1\leq j<k\leq N}(z_j-z_k)^2\,
\prod_{j=1}^N e^{-N V(z_j)}dz_1\ldots dz_N.
\end{equation}
We define the free energy as
\begin{equation}\label{free_general}
F_N:=\frac{1}{N^2}\ln \frac{Z_N}{Z^{\rm GUE}_N},
\end{equation}
where $Z^{\rm GUE}_N$ corresponds to the Gaussian Unitary Ensemble (GUE), with $V(M)=M^2/2$. In this case the partition function can be computed explicitly as a Selberg integral, cf. \cite{Ble}:
\begin{equation}\label{Zn0Selberg}
 Z_N^{\rm GUE}=\frac{(2\pi)^{N/2}}{N^{N^2/2}}\prod_{n=1}^N n!
\end{equation}

 If we consider a formal deformation of the
Gaussian model
\begin{equation}
V(M)=\frac{M^2}{2}+uW(M),
\end{equation}
where $u$ is a parameter, it is known that for small $u$ the free energy admits an asymptotic expansion
as $N\to\infty$, 
\begin{equation}\label{free_regular_general}
F_N(u)\sim \sum_{g=0}^{\infty} \frac{F^{(2g)}(u)}{N^{2g}}, \qquad N\to\infty,
\end{equation}
which is called the \emph{topological expansion}.
This asymptotic expansion in inverse powers of $N^2$ for the free energy is well known in the literature. For unitary ensembles it can be obtained from the Riemann--Hilbert formulation, as presented by Ercolani and McLaughlin in \cite{EMcL}, see also \cite{EMcLP} for a more detailed result on the coefficients in the asymptotic expansion for potentials $V(M)$ with a dominant even power.  Alternatively, it is possible to derive it via different techniques that are applicable to more general $\beta$-ensembles, see, for instance, the works of Borot and Guionnet \cite{BG1} and \cite{BG2}. 

It is a remarkable fact that the coefficients $F^{(2g)}(u)$ are generating functions for counting graphs 
on Riemannian surfaces of genus $g$. This important observation goes back  to the earlier physical papers 
of Bessis, Itzykson and Zuber \cite{BIZ}
 and Br\'ezin, Itzykson, Parisi and Zuber \cite{BIPZ}. For rigorous mathematical proofs, see  the works of 
Ercolani and McLaughlin \cite{EMcL}, Ercolani, McLaughlin and Pierce \cite{EMcLP}, and the monograph 
of Forrester \cite{For}.

The case when $W(M)$ is a cubic polynomial is especially interesting in this direction, since it gives a generating function for counting triangulations on surfaces. For cubic $W(M)$, however, the partition function, as defined in \eqref{ZN0}, diverges, and it is necessary to consider some kind of regularization. One way to achieve this is to define it using integration on a specially chosen contour $\Ga$ in the complex plane:
\begin{equation}\label{Zn}
Z_N(u)=\int_\Ga\ldots\int_\Ga \prod_{1\leq j<k\leq N}(z_j-z_k)^2\,
\prod_{j=1}^N e^{-N \left(\frac{z_j^2}{2}-uz_j^3\right)}dz_1\ldots dz_N,
\end{equation}
on which the integral converges.
Strictly speaking, this partition function does not correspond to any Hermitian random matrix model,
 but it serves, nevertheless, as a generating function for triangulations on Riemannian surfaces, see, e.g. \cite{BD}.

To choose an appropriate contour of integration $\Ga$, consider the three sectors on the complex plane,
\begin{equation}\label{sectors}
\begin{aligned}
S_0&=\Big\{z\in\C:\;\frac{5\pi}{6}< \arg z < \frac{7\pi}{6}\Big\}\,,\\
 S_1&=\Big\{z\in\C:\;\frac{\pi}{6}< \arg z < \frac{\pi}{4}\Big\}\,,\\
 S_2&=\Big\{z\in\C:\;-\frac{\pi}{4}< \arg z < -\frac{\pi}{6}\Big\}\,,
\end{aligned}
\end{equation}
see Fig.\ref{Figsectors}.

\begin{center}
\begin{figure}[h]
\begin{center}
\includegraphics{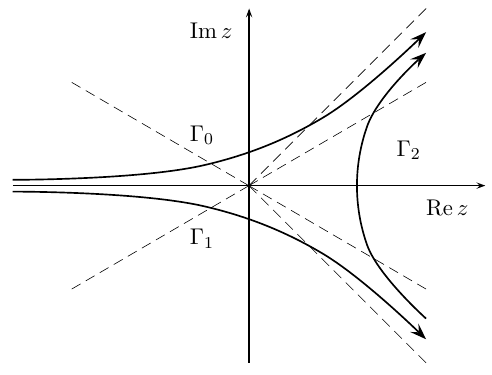}
\end{center}
  \caption[sectors ]{The sectors $S_0$, $S_1$, $S_2$ and the contours  $\Ga_0$, $\Ga_1$, $\Ga_2$.}
 \end{figure}
\label{Figsectors}
\end{center}

Then for any ray
\begin{equation}\label{ray}
R_{\theta}=\big\{z\in\C:\; \arg z =\theta \big\},
\end{equation}
lying in the sectors $S_0$, $S_1$, and $S_2$, the integral
\begin{equation}\label{integral}
\int_{R_{\theta}}z^k e^{-N \left(\frac{z^2}{2}-uz^3\right)}dz
\end{equation}
converges for any $k=0,1,\ldots$ and any $u\ge 0$. Clearly, it is also possible to take combinations of two such contours. In Fig. \ref{Figsectors} we consider contours consisting of two
rays joining the sectors $S_0$, $S_1$, and $S_2$, namely 
\begin{equation}\label{contours}
\Ga_0=R_\pi\cup R_{\pi/5},\qquad
\Ga_1=R_\pi\cup R_{-\pi/5},\qquad
\Ga_2=R_{-\pi/5}\cup R_{\pi/5},
\end{equation}
with orientation from $(-\infty)$ to $(\infty\,e^{\pi i/5})$ on $\Ga_0$,
from $(-\infty)$ to $(\infty\,e^{-\pi i/5})$ on $\Ga_1$, and
from $(\infty e^{-\pi i/5})$ to $(\infty\,e^{\pi i/5})$ on $\Ga_2$.  

More generally, following \cite{FIK}, see also \cite{DK} and \cite{BEH}, it is convenient to introduce a
linear combination of the contours $\Ga_0$ and $\Ga_1$. To that end,
let us fix some $\al\in\C$ and define 
\begin{equation}\label{Gammal}
\Ga=\al\Ga_0+(1-\al)\Ga_1,
\end{equation}
in the sense that for any $f(z)$ such that the integral of $f(z)$ along $\Gamma$ is well defined, we have
\begin{equation}\label{Gamma2}
\int_\Ga f(z)dz=\al\int_{\Ga_0}f(z)dz+(1-\al)\int_{\Ga_1}  f(z) dz.
\end{equation}

Observe that $\alpha=1$ corresponds to integration on $\Gamma_0$ only, and $\alpha=0$ corresponds to integration on $\Gamma_1$. The choice $\alpha=1/2$ would lead to integration on $\Gamma_2$, since by Cauchy's theorem, 
$$
\int_{\Ga_0} f(z)dz+\int_{\Ga_1}f(z)dz+\int_{\Ga_2}f(z)dz=0,
$$
for an analytic function $f(z)$.

With this choice of $\Ga=\Ga(\al)$, the integral
$Z_N(u)=Z_N(u;\al)$ in \eqref{Zn} is convergent for any $u\ge 0$. By
the Cauchy theorem, we have some flexibility in the choice of the
contours $\Ga_0$, $\Ga_1$ within the sectors $S_0$, $S_1$, $S_2$.

In our previous paper \cite{BD} we considered the large $N$ asymptotic behavior of the free energy 
\begin{equation}\label{Fn_cubic}
F_N(u)=\frac{1}{N^2}\ln \frac{Z_N(u)}{Z_N(0)},
\end{equation}
where $Z_N(u)$ is the partition function for the cubic model \eqref{Zn} in the interval $0\leq u<u_c$, where $u_c$ is the following critical value:
\begin{equation}\label{uc}
u_c=\frac{3^{1/4}}{18}.
\end{equation}

In this regular regime, the free energy admits an asymptotic expansion in powers of $N^{-2}$ of the form \eqref{free_regular_general}, see \cite{BD}. This expansion is uniform in the variable $u$ on any interval $[0, u_c-\delta]$ for $\delta>0$. Furthermore,
the functions $F^{(2g)}(u)$ do not depend on $\alpha$ and they admit an analytic
continuation to the disc $|u| < u_c$ in the complex plane.

The first two terms in the topological expansion \eqref{free_regular_general} for the cubic model can be written as power series in $u$, convergent for $|u|<u_c$, see \cite{BD}:
\begin{equation}\label{F0_regular}
F^{(0)}(u)=\frac{1}{2}\sum_{j=1}^{\infty}\frac{72^j \Gamma\left(\frac{3j}{2}\right)u^{2j}}{\Gamma(j+3)\Gamma\left(\frac{j}{2}+1\right)},
\end{equation}
and
\begin{equation}\label{F2_regular}
F^{(2)}(u)=\frac{5}{48}\sum_{j=1}^{\infty}\frac{72^j\Gamma(\frac{3j}{2})}{(3j+2)\Gamma(j+1)\Gamma(\frac{j}{2}+1)}
\, _3F_2\left(\begin{array}{l} -j+1,2,6\\ 5,-\frac{3j}{2}+1 \end{array};\frac{3}{2}\right) u^{2j},
\end{equation}
in terms of generalized hypergeometric functions, see \cite[Chapter 16]{DLMF}. Consequently, we can write an asymptotic expansion for the partition function for the cubic model in the regular regime:
\begin{equation}
 Z_N(u)=Z_N(0) e^{N^2 F^{(0)}(u)+F^{(2)}(u)}\left(1+\mathcal{O}(N^{-2})\right), \qquad N\to\infty.
\end{equation}

In this paper we analyze the behavior of the free energy and the partition function near the critical value $u=u_c$. Our analysis is based on the asymptotic formulas for recurrence coefficients for orthogonal polynomials. We consider monic orthogonal polynomials  $p_n(z)=z^n+\ldots$ on the contour $\Gamma$, defined in \eqref{Gammal}, with respect to the weight $w(z)$. The orthogonality condition is given by 
\begin{equation}\label{orthow}
\int_{\Gamma} p_k(z)p_j(z) w(z)dz=h_k\delta_{jk}, 
\end{equation}
where 
\begin{equation}\label{wz}
w(z)=e^{-NV(z;u)},\qquad V(z;u)=\frac{z^2}{2}-uz^3, 
\end{equation}
and $h_k=h_{N,k}(u)$ are normalizing constants. Assuming that the orthogonal polynomials $p_n(z)$ exist for $n=0,1,\ldots, N$, we have
\begin{equation}
Z_N(u)=\prod_{n=0}^{N-1} h_{N,n}(u),
\end{equation}
see for instance \cite{Ble}, and the orthogonal polynomials satisfy a three-term recurrence relation
\begin{equation}\label{TTRR}
zp_n(z)=p_{n+1}(z)+\beta_{N,n}p_k(z)+\gamma_{N,n}^2p_{n-1}(z),
\end{equation}
where $\beta_{N,n}=\beta_{N,n}(u)$, and
\begin{equation}\label{gamman_hn}
\gamma_{N,n}^2=\gamma_{N,n}^2(u)=\frac{h_{N,n}(u)}{h_{N,n-1}(u)}.
\end{equation}

%

Following the general theory, the existence of such a sequence of
orthogonal polynomials is a consequence of the Hankel determinant 
$D_n=\textrm{det}[\mu_{j+k}]_{0\leq j,k\leq n}$ being nonzero, where the moments are 
$\mu_j =\int_{\Gamma} z^jw(z)dz$, see for instance the monographs of Szeg\H{o}, \cite[Chapter II]{Sze} or Chihara,  \cite[\S 1.3]{Chi}. In this case, however,
existence is not guaranteed a priori, since the weight
function is not positive on $\Gamma$. However, one of the
consequences of the Riemann--Hilbert analysis will be the existence
of $p_n(z)$ for large enough $N$ and $n=N+\mathcal{O}(1)$.


The main goal of this paper is to obtain the asymptotic behavior of the partition function $Z_N(u)$ as $N\to\infty$ in the double scaling regime
\begin{equation}\label{ds_u}
|u-u_c|\leq CN^{-4/5}, 
\end{equation}
where $C>0$ is an arbitrary fixed number. To achieve this goal, we first obtain the asymptotic expansion of the recurrence coefficients $\beta_{N,N}(u)$ and $\gamma^2_{N,N}(u)$ in this double scaling regime, in Theorem \ref{Th1}. Then, in Theorem \ref{th_extended}, we extend these asymptotic expansions to a bigger neighborhood $u-u_c=\mathcal{O}(N^{-3/5})$, and simultaneously we extend the asymptotic expansion in the regular regime obtained in our previous paper \cite{BD} to $u-u_c\geq C_{\varepsilon}N^{-4/5+\varepsilon}$, for $\varepsilon>0$, where $C_{\varepsilon}>0$. Then, as a consequence of Theorem \ref{th_extended}, we have asymptotic expansions for the recurrence coefficients in the  overlapping regions $|u-u_c|=\mathcal{O}(N^{-3/5})$ and $0\leq u\leq u_c-C_{\varepsilon} N^{-4/5+\varepsilon}$, with $C_{\varepsilon}>0$. Finally, we integrate the Toda differential equation from $u=0$ to $u=u_c+\lambda N^{-4/5}$, with $|\lambda|\leq C$, and we obtain the double scaling asymptotic formula for the free energy in Theorem \ref{thfree}.


We remark that the appearance of solutions of Painlev\'e I in the double scaling asymptotics of the partition function near the critical point was predicted in the physics literature, see the papers by F. David \cite{David1991} and \cite{David1993}, and also in the work of M. Y. Mo, \cite{Mo}. 


The structure of the paper is as follows:
\begin{itemize}
\item In Section \ref{RHforPI} we recall known properties of the Painlev\'e I differential equation, and we present a Riemann--Hilbert problem for an associated function $\Psi$. A uniform asymptotic expansion for this function is also proved in Theorem \ref{Th_Phi}, a result that may be of independent interest.
\item In Section \ref{sec_main} we present the main results of the paper: asymptotic expansions for the recurrence coefficients corresponding to the orthogonal polynomials $p_n(z)$ and for the free energy in the double scaling regime.
\item In Section \ref{sec_eqmeasure} we analyze the equilibrium measure and its support close to the critical case. Similarly as in \cite{DK}, we need to construct a modified equilibrium measure, whose density becomes negative near one of the endpoints of the support.
\item In Section \ref{RHP} we apply the Deift--Zhou nonlinear steepest descent to the Riemann--Hilbert problem (see \cite{DZ} or the monograph \cite{Deift}), and as a result we deduce the asymptotic behavior of the recurrence coefficients $\ga^2_{N,N}(u)$ and $\be_{N,N}(u)$,
corresponding to the orthogonal polynomials of degree $n=N$ with respect to the
weight $e^{-NV(z;u)}$, in the double scaling regime and near the critical case. We also identify the subleading terms in these asymptotic expansions in terms of solutions of the Painlev\'e I differential equation, proving Theorem \ref{Th1}. Using the connection between the two double scaling formulations (in terms of $u-u_c$ and in terms of $n/N$, see \eqref{nNv}) we analyze the behavior of the recurrence coefficients $\ga^2_{N,n}(u_c)$ and $\be_{N,n}(u_c)$,
corresponding to the orthogonal polynomials of degree $n$ with respect to the
weight function $e^{-NV(z;u_c)}$, at the critical case. This proves Corollary \ref{Cor1}.

\item In Section \ref{sec_free} we integrate the Toda equation to obtain the asymptotic behavior of the free energy near the critical case. To prove Theorem \ref{thfree}, we
need to extend the regular and the double scaling regimes for the
asymptotics of the recurrence coefficients. This implies a small
modification of the corresponding local parametrices in the
Riemann--Hilbert problem. 
\end{itemize}



Note that both the orthogonal polynomials and the recurrence coefficients in the above relation depend
on the parameter $u$, and also on $n$ and $N$. When needed, we will denote these recurrence
coefficients by $\gamma^2_{N,n}(u)$ and $\beta_{N,n}(u)$, to
stress this dependence.

\section{Riemann--Hilbert problem for the Painlev\'e I equation}\label{RHforPI}

\subsection{The Painlev\'e I equation}
In order to state our results, we will need to work with solutions of the Painlev\'e I differential equation. This is a nonlinear second--order ordinary differential equation that in standard form reads
\begin{equation}\label{PI_general}
y''(\lambda)=6y(\lambda)^2+\lambda,
\end{equation}
see for instance \cite[Chapter 32]{DLMF}. Because of the Painlev\'e property, it is known, from the original work of Painlev\'e \cite[\S 17--\S19]{Pain}, see also \cite[Chapter 1, \S 1]{GLS}, that any solution of this differential equation is a meromorphic function in $\mathbb{C}$ with infinitely many double poles. Furthermore, if we denote the set of these poles by $\mathcal{P}$, then in a vicinity of any pole $\lambda_{j}\in\mathcal{P}$, the function $y(\lambda)$ can be expanded in Laurent series:
\begin{equation}\label{y_poles}
y(\lambda)=\frac{1}{(\lambda-\lambda_j)^2}-\frac{\lambda_j}{10}(\lambda-\lambda_j)^2-\frac{1}{6}(\lambda-\lambda_j)^3+C(\lambda-\lambda_j)^4+\mathcal{O}((\lambda-\lambda_j)^5),
\end{equation} 
where $C$ is an arbitrary constant.

We also note the following $\mathbb{Z}_5$-symmetry, that follows directly from the differential equation: if $y(\lambda)$ is a solution of \eqref{PI_general}, then the function
\begin{equation}
y^0(\lambda)=\frac{1}{\omega_5^2}\,y\left(\frac{\lambda}{\omega_5}\right), \qquad \omega_5=e^{2\pi i/5},
\end{equation}
is a solution as well. So we consider the rays
\begin{equation}\label{rays_PI}
\Sigma_k=\left\{\lambda\in\mathbb{C}:\arg\lambda=\frac{\pi}{5}+\frac{2(k-1)\pi}{5}, k=1,2,\ldots 5\right\},
\end{equation}
that delimit the sectors
\begin{equation}
\Omega_k=\left\{\lambda\in\mathbb{C}:\frac{\pi}{5}+\frac{2(k-1)\pi}{5}<\arg\lambda<\frac{\pi}{5}+\frac{2k\pi}{5}, k=1,2,\ldots 5\right\},
\end{equation}
see Figure \ref{Boutroux}. 

\begin{figure}
\centerline{
\includegraphics{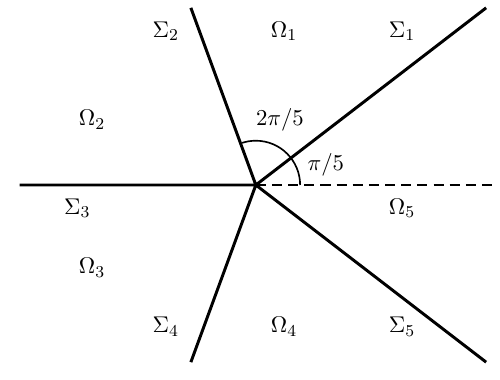}}
 \caption{Boutroux sectors in $\mathbb{C}$ for the solutions of Painlev\'e I.}
\label{Boutroux}
\end{figure}

These canonical sectors were originally considered by Boutroux \cite{Bou} in the analysis of solutions to Painlev\'e I. As proved in \cite{Bou}, for any $\lambda_j\in \mathbb C$, not located on any of the rays indicated in Figure \ref{Boutroux}, there exists a solution $y(\lambda)$ to Painlev\'e I that has a double pole at the point $\lambda_j$ and that is {\em tronqu\'ee}  in the direction of the semi axis $\Sigma_3$, in the following sense: for any $\varepsilon>0$ there exists $R>0$ such that in the region
\begin{equation}
\Lambda=\{\lambda\in\mathbb{C}: |\lambda|\geq R, \frac{3\pi}{5}+\varepsilon<\arg\lambda<\frac{7\pi}{5}-\varepsilon\}
\end{equation}
the function $y(\lambda)$ is free of poles. A similar result holds in the other sectors because of the $\mathbb{Z}_5$ symmetry shown before.

Moreover, the solution has the following asymptotic behavior in $\Lambda$:
\begin{equation}
y(\lambda)\sim \sqrt{-\frac{\lambda}{6}}\,\sum_{k=0}^{\infty} a_k(-\lambda)^{-5k/2},  \qquad |\lambda|\to\infty.
\end{equation}

Here the coefficients $a_k$ are given by the nonlinear recursion
\begin{equation}
a_0=1, \qquad a_{k+1}=\frac{25k^2-1}{8\sqrt{6}}a_k-\frac{1}{2}\sum_{m=1}^k a_ma_{k+1-m}, \qquad k\geq 0,
\end{equation}
so the first ones are
\begin{equation}
a_0=1, \qquad a_1=-\frac{1}{8\sqrt{6}}, \qquad a_2=-\frac{49}{768}.
\end{equation}

\subsection{The $\Psi$ function and its Riemann--Hilbert problem}\label{RHforPsi}
Kapaev in \cite{Kap} describes tronqu\'ee solutions of Painlev\'e I in terms of the following Riemann--Hilbert problem: introduce  a $2\times 2$ matrix valued function $\Psi(\z;\lambda,\alpha)$ such that
\begin{enumerate}
\item $\Psi$ is analytic on $\C\setminus \Ga_\Psi$, where
\begin{equation}\label{rhp5A}
\Ga_\Psi=\ga_1\cup\ga_2\cup\rho\cup \ga_{-2}\cup\ga_{-1}
\end{equation}
is the contour shown in Figure \ref{PI_jumps}.
\item On the contour $\Ga_\Psi$, oriented as in Figure \ref{PI_jumps}, the positive (from the left) and negative (from the right) limit values of $\Psi(\zeta)$, denoted $\Psi_{\pm}(\zeta)$ respectively, satisfy
\begin{equation}\label{rhp6A}
\Psi_+=\Psi_- J_\Psi
\end{equation}
where the jump matrices are indicated in Figure \ref{PI_jumps}:
\begin{equation}\label{rhp4A}
J_1=\begin{pmatrix}
1 & \al \\
0 & 1
\end{pmatrix},\quad
J_{-1}=\begin{pmatrix}
1 & 1-\al \\
0 & 1
\end{pmatrix},\quad
J_2=J_{-2}=\begin{pmatrix}
1 & 0 \\
1 & 1
\end{pmatrix},\quad
J_\rho=\begin{pmatrix}
0 & 1 \\
-1 & 0
\end{pmatrix}.
\end{equation}
\item As $\z\to\infty$, with fixed $\lambda$, we have the asymptotic series
\begin{equation}\label{rhp8A}
\Psi(\z;\lambda,\alpha)=\frac{\z^{\sigma_3/4}}{\sqrt{2}}\begin{pmatrix} 1 & -i\\ 1 & i\end{pmatrix}
\left(I+\frac{\Psi_1(\lambda,\alpha)}{\z^{1/2}}+\frac{\Psi_2(\lambda,\alpha)}{\z}+\mathcal{O}(\z^{-3/2})\right)e^{\theta(\z;\lambda)\sigma_3}.
\end{equation}
where $\Psi_1(\lambda,\alpha)$ is a diagonal matrix and the phase function is 
\begin{equation}\label{rhp9A}
\theta(\z;\lambda)=\frac{4}{5}\,\z^{5/2}+\lambda \z^{1/2},
\end{equation}
with fractional powers taken on the principal sheet and a cut on the negative half-axis. We use the standard notation for the Pauli matrix
\begin{equation}
\sigma_3=\begin{pmatrix} 1 & 0\\ 0 & -1\end{pmatrix}.
\end{equation}
\end{enumerate}

\begin{figure}
\centerline{
\includegraphics{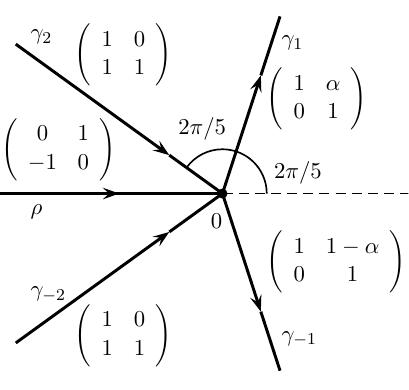}}
 \caption{The contour $\Gamma_{\Psi}$ and the jumps for the function $\Psi$ associated with the Painlev\'e I differential equation.}
\label{PI_jumps}
\end{figure}

In \cite{Kap} it is proved that this Riemann--Hilbert problem has a unique solution for large enough $|\lambda|$ in the sector $\Lambda$. Furthermore, the function $\Psi(\z;\lambda,\alpha)$ satisfies the following system of linear ODEs:
\begin{equation}\label{rhp2A}
\begin{aligned}
&\Psi_{\z}\Psi^{-1}=
\begin{pmatrix}
-(y_{\alpha})_{\lambda} & 2\z^2+2y_{\alpha}\z+\lambda+2y_{\alpha}^2 \\
2\z-2y_{\alpha} &(y_{\alpha})_{\lambda}
\end{pmatrix}\\
&\Psi_{\lambda}\Psi^{-1}=
\begin{pmatrix}
0 & \z+2y_{\alpha} \\
1 & 0
\end{pmatrix}.
\end{aligned}
\end{equation}

The compatibility condition of this Lax pair implies the Painlev\'e I equation \eqref{PI_general} for the function $y_{\alpha}(\lambda)$.  Here the subscripts $(\cdot)_{\lambda}$ and $(\cdot)_{\z}$ indicate differentiation with respect to these variables.

The parameter $\alpha$ that appears in the Riemann--Hilbert problem parametrizes a family of tronqu\'ee solutions of Painlev\'e I, with asymptotic behavior
\begin{equation}\label{asymp_ya}
y_{\alpha}(\lambda)\sim \sqrt{-\frac{\lambda}{6}}\,\sum_{k=0}^{\infty} a_k(-\lambda)^{-5k/2},  
\end{equation}
as $|\lambda|\to\infty$ in $\Lambda$. The dependence on $\alpha$ appears in exponentially small corrections to this asymptotic behavior.  These extra terms are computed by Kapaev in \cite[Theorem 2.2]{Kap}. If $\arg \lambda\in[3\pi/5,\pi]$, then the large $|\lambda|$ asymptotics of $y_{\alpha}(\lambda)$ is given by
\begin{equation}\label{asymp_alpha1}
y_{\alpha}(\lambda)=y_1(\lambda)+\frac{i(1-\alpha)}{\sqrt{\pi}\,2^{11/8}(-3\lambda)^{1/8}}e^{-\frac{2^{11/4}3^{1/4}}{5}(-\lambda)^{5/4}}\left(1+\mathcal{O}(\lambda^{-3/8})\right), \qquad |\lambda|\to\infty,
\end{equation}
where $y_{1}(\lambda)$ behaves as \eqref{asymp_ya}. For higher order corrections in the latter formula, $k$-instanton terms, see the paper \cite{GIKM} and references therein.

A similar result holds in the sector $\arg\lambda\in[\pi,7\pi/5]$, with respect to $y_0(\lambda)$ instead of $y_1(\lambda)$. As a consequence, \cite[Corollary 2.4]{Kap} establishes that if $\arg\lambda\in[3\pi/5,7\pi/5]$, then
\begin{equation}\label{asymp_alpha10}
y_{0}(\lambda)-y_1(\lambda)=\frac{i}{\sqrt{\pi}\,2^{11/8}(-3\lambda)^{1/8}}e^{-\frac{2^{11/4}3^{1/4}}{5}(-\lambda)^{5/4}}\left(1+\mathcal{O}(\lambda^{-3/8})\right), \qquad |\lambda|\to\infty.
\end{equation}

Furthermore, the cases $\alpha=0$ and $\alpha=1$ are special in the sense that, as $|\lambda|\to\infty$, we have 
\begin{equation}
\begin{aligned}
y_0(\lambda)&= \sqrt{-\frac{\lambda}{6}}+\mathcal{O}(\lambda^{-2}), \qquad \varepsilon+\frac{3\pi}{5}<\arg\lambda<\frac{11\pi}{5}-\varepsilon,\\
y_1(\lambda)&= \sqrt{-\frac{\lambda}{6}}+\mathcal{O}(\lambda^{-2}), \qquad \varepsilon-\frac{\pi}{5}<\arg\lambda<\frac{7\pi}{5}-\varepsilon,
\end{aligned}
\end{equation}
for arbitrary $\varepsilon>0$. This means that these are {\em tritronqu\'ee} solutions (asymptotically free of poles in four out of the five canonical sectors), again after Boutroux \cite{Bou}, see also \cite{JosKit}. 

The first two terms $\Psi_1(\lambda,\alpha)$ and $\Psi_2(\lambda,\alpha)$ in the asymptotic expansion \eqref{rhp8A} are fixed and can be written as follows, cf. \cite{DK}:
\begin{equation}\label{Psi1Psi2}
\begin{aligned}
\Psi_1(\lambda,\alpha)&=-\mathcal{H}_{\alpha}(\lambda)\sg_3=\begin{pmatrix}-\mathcal{H}_{\alpha}(\lambda) & 0\\
0 & \mathcal{H}_{\alpha}(\lambda)\end{pmatrix},\\
\Psi_2(\lambda,\alpha)&=\frac{1}{2}\mathcal{H}_{\alpha}(\lambda)^2 I+\frac{1}{2}y_{\alpha}(\lambda)\begin{pmatrix}0 &-i\\ i &0\end{pmatrix}=
\frac{1}{2}\begin{pmatrix}\mathcal{H}_{\alpha}(\lambda)^2 & -iy_{\alpha}(\lambda)\\
iy_{\alpha}(\lambda) & \mathcal{H}_{\alpha}(\lambda)^2\end{pmatrix}.
\end{aligned}
\end{equation}
Here 
\begin{equation}
\mathcal{H}_{\alpha}(\lambda)=\frac{1}{2}(y_{\alpha}'(\lambda))^2-2y_{\alpha}^3(\lambda)-y_{\alpha}(\lambda)\lambda
\end{equation}
is the Hamiltonian corresponding to Painlev\'e I. Consequently, given a fixed value of $\alpha$, the function $y_{\alpha}(\lambda)$ can be written in terms of the solution of this Riemann--Hilbert problem:
\begin{equation}
y_{\alpha}(\lambda)=2i(\Psi_2(\lambda))_{12}.
\end{equation} 

\begin{remark}
The parameter $\alpha$, that appears again in the jumps of the Riemann--Hilbert problem, can be naturally related to the Stokes multipliers for Painlev\'e I, as formulated by Kapaev, namely $\alpha=-i s_1$ or $\alpha=1+is_{-1}$. Note however that the function $\Psi(\z;\lambda,\alpha)$ that we just defined is not exactly the same that is used in \cite{Kap}. If we denote this last one by $\Psi^{(0)}$, we have the relation
\begin{equation}
\Psi(\z;\lambda,\alpha)=\Psi^{(0)}(\z;\lambda,\alpha)\begin{pmatrix} 1 & 0 \\ 0 & -i \end{pmatrix},
\end{equation}
and the jump matrices are related as follows:
\begin{equation}
J_{\Psi}=\begin{pmatrix} 1 & 0 \\ 0 & i \end{pmatrix}J_{\Psi^{(0)}}\begin{pmatrix} 1 & 0 \\ 0 & -i \end{pmatrix}.
\end{equation}
\end{remark}

\subsection{Large $|\lambda|$ asymptotics for the solution of the Painlev\'e I Riemann--Hilbert problem}
The asymptotic expansion \eqref{rhp8A} holds for large $\z$ and fixed $\lambda$, but for the analysis of the free energy later on, we need to extend it to cover the case when $(-\lambda)\to\infty$, $\lambda\in\mathbb{R}$, simultaneously. To this end, we define the following function:
\begin{equation}\label{rhp11A}
\Phi(\z;\lambda,\alpha)=(-\lambda)^{-\sg_3/8}\Psi(\z(-\lambda)^{1/2};\lambda,\alpha).
\end{equation}
Then $\Phi$ solves the following RH problem:
\begin{enumerate}
\item $\Phi$ is analytic on $\C\setminus \Ga_\Psi$.
\item On $\Ga_\Psi$, $\Phi$ has the jumps,
\begin{equation}\label{rhp12A}
\Phi_+=\Phi_- J_\Psi.
\end{equation}
\item As $\z\to\infty$, $\Phi$ expands in the asymptotic series
\begin{equation}\label{rhp13A}
\Phi(\z;\lambda,\alpha)\sim \frac{\z^{\sg_3/4}}{\sqrt 2}
\begin{pmatrix}
1 & -i \\
1 & i
\end{pmatrix}
\left(I+\sum_{k=1}^\infty \frac{\Psi_k(\lambda,\alpha)}{(-\lambda)^{k/4}\z^{k/2}}\right)e^{(-\lambda)^{5/4}\theta_0(\z)\sg_3},
\end{equation}
where now
\begin{equation}\label{rhp14A}
\theta_0(\z)=\frac{4}{5}\,\z^{5/2}-\z^{1/2}.
\end{equation}
\end{enumerate}

The first equation in \eqref{rhp2A} gives the differential equation for $\Phi$:
\begin{equation}\label{rhp15A}
\Phi_{\z}\Phi^{-1}=
(-\lambda)^{5/4}A
\end{equation}
where
\begin{equation}\label{rhp16A}
A=A(\z;\lambda,\alpha)=\begin{pmatrix}
-(y_{\alpha})_{\lambda}(-\lambda)^{-3/4} & 2\z^2+2y_{\alpha} \z(-\lambda)^{-1/2}-1+2(y_{\alpha})^2(-\lambda)^{-1} \\
2\z-2y_{\alpha}(-\lambda)^{-1/2} & (y_{\alpha})_{\lambda}(-\lambda)^{-3/4}
\end{pmatrix},
\end{equation}
where again $(\cdot)_{\lambda}$ indicates differentiation with respect to this variable. 

It is possible to derive a semiclassical solution to this system for large $|\lambda|$ by using the asymptotic behavior of $y_{\alpha}(\lambda)$. As $\lambda\to-\infty$, the matrix $A$
has the following limit:
\begin{equation}\label{rhp17A}
\begin{aligned}
\lim_{\lambda\to -\infty} A(z;\lambda,\alpha)&=A_\infty(\z)=\begin{pmatrix}
0 & 2\z^2+\frac{2\z}{\sqrt 6}-\frac {2}{3} \\
2(\z-\frac{1}{\sqrt 6}) & 0
\end{pmatrix}\\
&=2\left(\z-\frac{1}{\sqrt 6}\right)\begin{pmatrix}
0 & \z+\frac{2}{\sqrt 6} \\
1 & 0
\end{pmatrix}.
\end{aligned}
\end{equation}
The eigenvalues of $A_\infty$ are
\begin{equation}\label{rhp18A}
\z_{1,2}=\pm 2\left(\z-\frac{1}{\sqrt 6}\right)\left(\z+\frac{2}{\sqrt 6}\right)^{1/2}
\end{equation}
and we introduce the $g$-function as
\begin{equation}\label{rhp18}
g(\z)
=\frac{4}{5}\left(\z+\frac{2}{\sqrt 6}\right)^{5/2}-\frac{2\sqrt 6}{3}\left(\z+\frac{2}{\sqrt 6}\right)^{3/2},
\end{equation}
in such a way that
\begin{equation}
 g'(\z)=2\left(\z-\frac{1}{\sqrt 6}\right)\left(\z+\frac{2}{\sqrt 6}\right)^{1/2}.
\end{equation}
 
As $\z\to\infty$, $g(\z)$ has the following power series expansion:
\begin{equation}\label{rhp19A}
g(\z)
=\frac{4}{5}\,\z^{5/2}-\z^{1/2}-\frac{\sqrt{6}}{9}\,\z^{-1/2}+\frac{1}{24}\,\z^{-3/2}-\frac{\sqrt{6}}{180}\,\z^{-5/2}+\ldots
\end{equation}

Hence, the first two terms in the expansion coincide with the phase function $\theta_0(\lambda;\alpha)$ in \eqref{rhp14A}, and the subsequent terms are small as $\z\to\infty$, so we can rewrite the asymptotic expansion \eqref{rhp13A} as follows:
\begin{equation}\label{rhp13g}
\Phi(\z;\lambda,\alpha)\sim \frac{\z^{\sg_3/4}}{\sqrt 2}
\begin{pmatrix}
1 & -i \\
1 & i
\end{pmatrix}
\left(I+\sum_{k=1}^\infty \frac{\Phi_k(\lambda,\alpha)}{(-\lambda)^{k/4}\z^{k/2}}\right)e^{(-\lambda)^{5/4}g(\z)\sg_3},
\end{equation}
with some modified coefficients $\Phi_k(\lambda,\alpha)$. 

Regarding a large $|\lambda|$ asymptotic expansion for the function $\Phi(\z;\lambda,\alpha)$, uniform in $\z$, we have the following result:
\begin{theorem}\label{Th_Phi}
Fix $\alpha\in\mathbb{C}$, then $\Phi(\z;\lambda,\alpha)$ has the following asymptotic behavior as $(-\lambda)\to\infty$:
\begin{equation}\label{asymp_Phi_final}
\Phi(\z;\lambda,\alpha)=
\left(I+\mathcal{O}\left(\frac{1}{(-\lambda)^{5/4}(1+|\z|)}\right)\right)
\frac{(\z-\z_0)^{\sg_3/4}}{\sqrt{2}}\begin{pmatrix} 1& -i\\ 1 & i \end{pmatrix}
e^{(-\lambda)^{5/4}g(\z)\sg_3},
\end{equation}
where 
\begin{equation}\label{zeta0}
\z_0=-\frac{2}{\sqrt{6}},
\end{equation}
and $g(\z)$ is given by \eqref{rhp18}. The estimate holds uniformly for $\z\in\mathbb{C}\setminus D(\z_0,\varepsilon)$, for any $\varepsilon>0$.

\end{theorem}

In the proof of this result we use the Riemann--Hilbert approach to analyze $\Phi(\z;\lambda,\alpha)$, and we consider the sequence of transformations of the Deift--Zhou steepest descent method. The proof can be found in Appendix \ref{ApA}.

\section{Statement of main results}\label{sec_main}

We prove the following results:

\begin{theorem}\label{Th1}
Fix $\alpha\in\mathbb{C}$, and suppose that $\lambda\in\mathbb{R}$ is not a pole of the Painlev\'e I function $y_{\alpha}(\lambda)$. Under the double
scaling relation
\begin{equation}\label{dsc_u}
N^{4/5}(u-u_c)=c_1\lambda, \qquad c_1=2^{-12/5}3^{-7/4},
\end{equation}
the recurrence coefficients for the polynomials orthogonal with respect to the weight $e^{-NV(z;u)}$, with $V(z;u)=\tfrac{z^2}{2}-u z^3$, satisfy
\begin{equation}\label{asympgnbn2N}
\begin{aligned}
\ga_{N,N}^2(u)&\sim \sqrt{3}+\sum_{k=1}^{\infty}\frac{1}{N^{2k/5}}\, p_{2k}(\lambda),\\
\be_{N,N}(u) & \sim
3^{1/4}(\sqrt{3}-1)+\sum_{k=1}^{\infty}\frac{1}{N^{2k/5}}\,q_{2k}(\tilde{\lambda}).
\end{aligned}
\end{equation}
for some coefficients $p_{2k}(\lambda)$ and $q_{2k}(\tilde{\lambda})$, where $\tilde \lambda$ is a
shifted variable defined in terms of $\lambda$ as
\begin{equation}
\tilde{\lambda}=N^{4/5}\varphi\left(\varphi^{-1}(\lambda N^{-4/5})+\frac{1}{2N}\right),
\qquad \varphi(t)=\frac{\sqrt{1+t}-1}{c_2}.
\end{equation}

Furthermore, we have
\begin{equation}
p_2(\lambda)=-2^{4/5}3^{1/2} y_{\alpha}(\lambda),
\end{equation}
in terms of the solution to Painlev\'e I, and $q_{2}(\lambda)=3^{-1/4}p_2(\lambda)$.
\end{theorem}

We can prove a similar result with the double scaling in terms of $n/N$. Note that if we apply the change of variables
\begin{equation}\label{zxi}
z=\frac{u}{u_c}\xi,
\end{equation}
then we obtain $NV(z;u_c)=nV(\xi;u)$, where $n=Nu^2/u_c^2$. This leads to an alternative double scaling for the ratio $n/N$, at the critical point $u=u_c$:
\begin{equation}
\frac{n}{N}=\frac{u^2}{u_c^2}=(1+c_2\lambda N^{-4/5})^2,
\end{equation}
where $c_2=c_1/u_c=2^{-7/5}$ is a new constant. Consequently, we define a new variable $v$ as follows:
\begin{equation}\label{nNv}
\frac{n}{N}=1+vN^{-4/5}.
\end{equation}

In this setting, we fix $u=u_c$ and we let both $n,N\to\infty$, in such a way that $v$ in the previous formula remains fixed. We have the following result:
\begin{corollary}\label{Cor1}
Fix $u=u_c$, and let $n$ and $N$ satisfy the double scaling relation \eqref{nNv}, with $v\in\mathbb{R}$ bounded. Then the recurrence coefficients satisfy
\begin{equation}
\begin{aligned}
\ga_{N,n}^2&\sim \ga^2_c+\sum_{k=1}^{\infty}\frac{1}{N^{2k/5}}\, \hat{p}_{2k}(v),\\
\be_{N,n}& \sim \be_c+\sum_{k=1}^{\infty}\frac{1}{N^{2k/5}}\,\hat{q}_{2k}(\tilde{v}),
\end{aligned}
\end{equation}
for some coefficients $\hat{p}_{2k}(v)$ and
$\hat{q}_{2k}(\tilde{v})$, and with the same values of $\ga^2_c$ and $\be_c$ as in the previous theorem. The variable $v$ is defined in \eqref{nNv} and $\tilde v$ i\newcommand{\ee}{\textrm{e}}
s a shifted variable defined in terms of $v$ as
\begin{equation}
\tilde v=v+\frac{N^{-1/5}}{2},
\end{equation}
in such a way that
\begin{equation}
\frac{n}{N}+\frac{1}{2N}=1+\tilde vN^{-4/5}.
\end{equation}
\end{corollary}

\begin{remark}
As noticed in \cite{DK}, the terms of order $N^{-1/5}$ cancel out in the asymptotic expansion of the recurrence coefficients. Using the shifted variables $\tilde{\lambda}$ and $\tilde{v}$ before, this is true in general for higher order odd terms.
\end{remark}

In order to formulate our main result on the asymptotic behavior of the partition function, we need to extend the asymptotic expansion of the recurrence coefficients both from the regular and the double scaling regime, in order to be able to integrate the Toda equation. Namely, we have the following result:
\begin{theorem}\label{th_extended}
The asymptotic expansions \eqref{asympgnbn2N} for $\gamma^2_{N,N}(u)$ and $\beta_{N,N}(u)$ hold for $N^{4/5}(u-u_c)=c_1\lambda$, where $\lambda=\mathcal{O}(N^{4/25-\varepsilon_1})$, for $\varepsilon_1>0$. Simultaneously, for any $K\geq 1$ and $u-u_c=\mathcal{O}(N^{-4/5+\varepsilon_2})$, for $\varepsilon_2>0$, we have the truncated regular expansion
\begin{equation}
\begin{aligned}
\gamma^2_{N,N}(u)&= g_0(u)+\sum_{k=1}^{K}\frac{g_{2k}(u)}{N^{2k}}+\mathcal{O}(N^{-\frac{2}{5}-\frac{5K+4}{2}\varepsilon_2}),\\
\beta_{N,N}(u)&= b_0(u)+\sum_{k=1}^{K}\frac{b_{2k}(u)}{N^{2k}}+\mathcal{O}(N^{-\frac{2}{5}-\frac{5K+4}{2}\varepsilon_2}),
\end{aligned}
\end{equation}
with the coefficients $g_{2k}(u)$ and $b_{2k}(u)$ from the regular regime, see \cite{BD}.
\end{theorem}

\begin{remark}
 The estimate of the error term comes from the extended regular regime $u-u_c=\mathcal{O}(N^{-4/5+\varepsilon_2})$, together with \cite{BD}, where it is shown that
\begin{equation}
 g_{2k}(u),b_{2k}(u)=\mathcal{O}((u_c-u)^{\frac{1}{2}-\frac{5k}{2}}),
\end{equation}
for $k\geq 1$. The zeroth terms satisfy
\begin{equation}
 g_0(u)=\sqrt{3}+\mathcal{O}((u_c-u)^{\frac{1}{2}}), \qquad
 b_0(u)=3^{1/4}(\sqrt{3}-1)+\mathcal{O}((u_c-u)^{\frac{1}{2}}).
\end{equation}

\end{remark}

As a consequence of this theorem, in the intermediate region $\lambda=\mathcal{O}(N^{\varepsilon_3})$, with $0<\varepsilon_3<4/25$, we can match both asymptotic expansions and integrate to obtain the behavior of the free energy.

Using this information, we can prove our main result on the asymptotic behavior of the partition function:
\begin{theorem}\label{thfree}
Given $\varepsilon>0$ and $\delta>0$, consider the double scaling regime \eqref{dsc_u}, $N^{4/5}(u-u_c)=c_1\lambda$ and fix a neighborhood in the complex plane $D_R=\{\lambda\in\mathbb{C}:|\lambda|<R\}$. Let $\{\lambda_{\alpha,j}\}_{j=1}^J$ be the set of poles of $y_{\alpha}(\lambda)$ in $D_R$. The partition function $Z_N(u)$ can be written in the following way:
\begin{equation}
Z_N(u)=Z_N^{\operatorname{reg}}(u)Z_N^{\operatorname{sing}}(\lambda)\left(1+\mathcal{O}(N^{-\varepsilon})\right),
\end{equation}
for $\lambda\in D_R\setminus \cup_j D(\lambda_j,\delta)$. Here the regular part is
\begin{equation}\label{ZNreg}
Z_N^{\operatorname{reg}}(u)=e^{N^2 [A+B(u-u_c)+C(u-u_c)^2]+D},
\end{equation}
where the constants $A$, $B$, $C$ and $D$ are explicit:
\begin{equation}
A=F^{(0)}(u_c), \qquad 
B=F^{(0)'}(u)\Big\vert_{u=u_c}, \qquad 
C=\frac{1}{2} F^{(0)''}(u)\Big\vert_{u=u_c}, 
\end{equation}
where $F^{(0)}(u)$ is given by \eqref{F0_regular}, and
\begin{equation}
D=\left[F^{(2)}(u)+\frac{1}{48}\ln (u_c-u)\right]\Bigg\vert_{u=u_c},
\end{equation}
where $F^{(2)}(u)$ is given by \eqref{F2_regular}. The singular part of the partition function is
\begin{equation}\label{ZNsing}
Z_N^{\operatorname{sing}}(\lambda)=e^{-Y_{\alpha}(\lambda)},
\end{equation}
where $Y_{\alpha}(\lambda)$ solves the differential equation
\begin{equation}\label{Yy_ODE}
Y''_{\alpha}(\lambda)=y_{\alpha}(\lambda),
\end{equation}
with boundary condition
\begin{equation}\label{asymp_Y}
Y_{\alpha}(\lambda)=\frac{2\sqrt{6}}{45}(-\lambda)^{5/2}-\frac{1}{48}\log(-\lambda)+\mathcal{O}((-\lambda)^{-5/2}), \qquad (-\lambda)\to\infty.
\end{equation}
\end{theorem}

\begin{figure}[h!]
\centerline{
\includegraphics{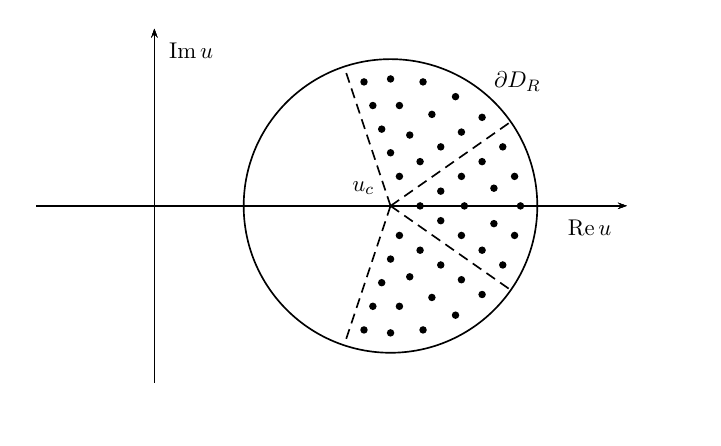}}
 \caption{Illustration of the domain used in Theorem \ref{thfree}. Dots represent poles of $y_{\alpha}(\lambda)$, schematically. Here $\partial D_R=\{u\in\mathbb{C}: N^{4/5}|u-u_c|=c_1 R\}$.}
\label{free_pic}
\end{figure}

As a consequence we can relate the poles of $y_{\alpha}(\lambda)$ with zeros of the partition function, cf. \cite{David1993}:
\begin{corollary}
Assume that there are no poles of $y_{\alpha}(\lambda)$ on the boundary of the disc $\partial D_R$, then for large $N$, the partition function $Z_N(u)$ has exactly $J$ zeros $\tau_j(N)$ in $D_R$, and $\tau_j(N)\to\lambda_j$ as $N\to\infty$.
\end{corollary}

\begin{proof}
Taking into account the behavior of $y_{\alpha}(\lambda)$ near the poles, see \eqref{y_poles}, we obtain by integration
\begin{equation}\label{Y_zeros}
Y_{\alpha}(\lambda)=-\log(\lambda-\lambda_j)+\mathcal{O}((\lambda-\lambda_j)^4), \qquad j=1,2\ldots, J,
\end{equation}
so on the boundary of $D(\lambda_j,\delta)$ we have 
\begin{equation}\label{Y_poles}
Z_{N}^{\operatorname{sing}}(\lambda)=(\lambda-\lambda_j)(1+\mathcal{O}(\lambda-\lambda_j)).
\end{equation}

By the argument principle, $Z_{N}^{\operatorname{sing}}(\lambda)$ has exactly one zero in $D(\lambda_j,\delta)$.
\end{proof}

In Figure \ref{free_pic} we illustrate the setting of the theorem schematically. Note that the variable used is $u$ instead of $\lambda$, so the boundary of the disc is the set $\{u\in\mathbb{C}:N^{4/5}(u-u_c)=c_1\lambda\}$, and it will shrink with $N$ for fixed $\lambda$.

\section{The equilibrium measure}\label{sec_eqmeasure}
As shown in \cite{EMcL}, a key element in the analysis of the partition function and free energy is the equilibrium measure in the external field $V(z;u)$. In the cubic case this equilibrium measure can be written explicitly, as we show next.

\subsection{Support of the equilibrium measure}
When $0\leq u<u_c$, we know from results in \cite{BD} that the equilibrium
measure is supported on an interval $[a,b]$ of the real axis, and
the density is
\begin{equation}\label{rhoz}
\varrho_u(z)=\frac{1}{2\pi} \sqrt{(z-a)(b-z)}(1-3uz-3ux),
\end{equation}
where both $a$ and $b$ depend on $u$, and $x=(a+b)/2$. 

The parameter $\tau=ux$ satisfies the following cubic equation, see \cite{BD}:
\begin{equation}\label{cubictau}
18\tau^3-9\tau^2+\tau-6u^2=0.
\end{equation}

Denote $s=108\sqrt{3}\,u^2$, then the cubic equation becomes
\begin{equation}\label{cubictaus}
18\tau^3-9\tau^2+\tau-\frac{\sqrt{3}}{54}s=0,
\end{equation}
and it has three solutions, that behave as follows as $s\to 0$:
\begin{equation}
\begin{aligned}
\tau_1(s)&=\frac{1}{3}+\frac{\sqrt{3}}{54}s+\mathcal{O}(s^2),\qquad
\tau_2(s)=\frac{1}{6}-\frac{\sqrt{3}}{27}s+\mathcal{O}(s^2),\qquad
\tau_3(s)=\frac{\sqrt{3}}{54}s+\frac{1}{108}s^2+\mathcal{O}(s^3).
\end{aligned}
\end{equation}

Furthermore, the discriminant of the cubic equation \eqref{cubictaus} is
$\Delta(s)=9(1-s^2)$. When $s=1$, the cubic equation \eqref{cubictaus}
has a double root and a single root:
\begin{equation}
\tau_1(1)=\frac{1}{6}+\frac{\sqrt{3}}{9}, \qquad
\tau_2(1)=\tau_3(1)=\frac{1}{6}-\frac{\sqrt{3}}{18},
\end{equation}
and similarly when $s=-1$:
\begin{equation}
\tau_1(-1)=\tau_2(-1)=\frac{1}{6}+\frac{\sqrt{3}}{18}, \qquad
\tau_3(-1)=\frac{1}{6}-\frac{\sqrt{3}}{9}.
\end{equation}

Since we are interested in the solution $x(u)=\tau(s)/u$ that is bounded near $u^2=s=0$, we need to 
choose the solution $\tau_3(s)$.

For convenience, we make the following linear change of variables (depending on $u$):
\begin{equation}\label{zetaz}
\zeta=\frac{2z-a-b}{b-a},
\end{equation}
so that the interval $[a,b]$ is mapped to
$[-1,1]$. Then the equilibrium measure becomes
\begin{equation}\label{rhou}
\varrho_u(\z)d\z=-\frac{3(b-a)^3u}{16\pi}\sqrt{1-\zeta^2}\left(\zeta-\frac{1-6ux}{3uy}\right)d\z,
\end{equation}
where $y=(b-a)/2$. This parameter satisfies the equation
\begin{equation}
y^2=\frac{4}{1-6ux},
\end{equation}
so
\begin{equation}
y=\frac{2}{\sqrt{1-6ux}}=\frac{2}{\sqrt{1-6\tau}},
\end{equation}
since we assume that $b>a$ and therefore $y>0$.

The extra root of $\varrho_u(\z)$ is
\begin{equation}\label{z0}
\z_0=\frac{1-6ux}{3uy}=\frac{1}{6u}(1-6\tau)^{3/2},
\end{equation}
 which is outside the interval $[-1,1]$ when $u<u_c$ 
(actually $\z_0>1$), and coalesces with the right endpoint when $u=u_c$, see \cite{BD}. 

Near the critical point $s=1$, we write $\Delta s=1-s$ and by perturbation we obtain
\begin{equation}
\tau_{3\pm}(s)=\frac{1}{6}-\frac{\sqrt{3}}{18}\pm
\frac{\sqrt{2}}{18}(\Delta s)^{1/2}+\frac{\sqrt{3}}{162}\Delta s+\mathcal{O}((\Delta s)^{3/2}).
\end{equation}

For $y$, we have the solutions
\begin{equation}
y_3(s)=2\cdot 3^{1/4}+2^{1/2}
3^{-1/4}(\Delta s)^{1/2}+\frac{11\cdot
3^{1/4}}{18}\Delta s+\mathcal{O}((\Delta s)^{3/2}).
\end{equation}

Regarding the double root, we have
\begin{equation}\label{z0s}
\zeta_0=1+\frac{\sqrt{6}}{2}(\Delta s)^{1/2}+\frac{7}{12}\Delta s+\mathcal{O}((\Delta s)^{3/2}),
\end{equation}
taking the solution $\tau_{3+}(s)$.

Note that this is consistent with the expected behavior: if $u<u_c$
then $s<1$, and $\zeta_0>1$, so the double root is outside the
interval where the equilibrium measure is supported. If $u>u_c$ then
$s>1$, and both the endpoints and the double root become complex. A more complete picture of the different cases that can occur in this cubic model, using the notion of $S$-curves in the complex plane and numerical computations, is described in \cite[\S 4]{AMAM}. Note also that the dependence of the endpoints on the parameter $s$ is not analytic near the critical case.

\subsection{Modified equilibrium measure at the critical case}
When $u=u_c$, the equilibrium measure is supported on the interval $[a_c,b_c]=[3^{3/4}-3^{5/4},3^{3/4}+3^{1/4}]$, and in the new variable $\z$, see \eqref{zetaz}, the polynomial $V(z;u)=\frac{z^2}{2}-uz^3$ becomes
\begin{equation}\label{Vzeta}
V(\zeta;u)=V_{cr}(\zeta)+(u-u_c)V^o(\zeta).
\end{equation}
where
\begin{equation}\label{Vzeta1}
V_{cr}(\zeta)=\frac{(2\z+\sqrt{3}-1)^2(-2\z+2\sqrt{3}+1)}{6}
\end{equation}
and
\begin{equation}\label{Vzeta2}
V^o(\zeta)=-3^{3/4}(2\z+\sqrt{3}-1)^3.
\end{equation}
The equilibrium measure at the critical time has the density, from \eqref{rhou}:
\begin{equation}
\varrho_{cr}(x)dx=\frac{2}{\pi }(1-x)\sqrt{1-x^2}\,dx,\qquad -1\le x\le
1.
\end{equation}
We observe that this is indeed a probability density function, but the equilibrium measure is not regular, since its density vanishes with a $3/2$ exponent at the right endpoint. In this respect, the analysis is very similar to the one carried out in \cite{DK}, but without the symmetry around the origin present in that case.

The resolvent of $\varrho_{cr}(x)$ is
\begin{equation}
\begin{aligned}
\om_{cr}(\zeta)=\int_{-1}^1 \frac{\varrho_{cr}(x)dx}{\z-x}\,,
\qquad \zeta\in\mathbb{C}\setminus [-1,1]
\end{aligned}
\end{equation}
and it solves the equation
\begin{equation}
\begin{aligned}
\om_{cr+}(x)+\om_{cr-}(x)=V'_{cr}(x)=-4x^2+4x+2\,,\qquad x\in(-1,1).
\end{aligned}
\end{equation}

It is equal to
\begin{equation}\label{omcr}
\begin{aligned}
\om_{cr}(\zeta)=\frac{V'_{cr}(\z)}{2}+\pi i\varrho_{cr}(\z)=-2\z^2+2\z+1+2(\zeta+1)^{1/2}(\zeta-1)^{3/2}.
\end{aligned}
\end{equation}

Let us extend the function $\varrho_{cr}(x)$ to the complex plane as
\begin{equation}
\begin{aligned}
\varrho_{cr}(\z)=\frac{2i}{\pi }(\z+1)^{1/2}(\z-1)^{3/2},\qquad
\z\in\C\setminus[-1,1],
\end{aligned}
\end{equation}
with a cut on $[-1,1]$, so that
\begin{equation}
\begin{aligned}
\varrho_{cr}(x+i0)=\frac{2}{\pi }(x+1)^{1/2}(1-x)^{3/2},\qquad
x\in(-1,1).
\end{aligned}
\end{equation}
In what follows the following function will be important:
\begin{equation}\label{phicr1}
\phi_{cr}(\z)=-2\pi i\int_{1}^{\zeta}
\varrho_{cr}(s)ds=4\int_{1}^{\zeta} (s+1)^{1/2}(s-1)^{3/2}ds,
\qquad \z\in\C\setminus (-\infty,1],
\end{equation}
where the integration goes over the segment $[1,\z]$ on the complex plane. It can be integrated explicitly as
\begin{equation}\label{phicr1exp}
\phi_{cr}(\z)=\frac{2\sqrt{\zeta^2-1}(\zeta-2)(2\zeta+1)}{3}+2\log\left(\z+\sqrt{\z^2-1}\right).
\end{equation}

We plot the level curves of the function $\Re \phi_{cr}(\z)$
in Figure \ref{GaY}, and based on this information we take the following contour $\Ga_Y$ for the
Riemann--Hilbert analysis: the real axis from $-\infty$ to $\zeta=1$
and then a combination of steepest descent path of $\textrm{Re}\, \phi_{cr}(\zeta)$ into the upper and lower half plane, on which $\Im \phi_{cr}(\z)=0$, see Figure \ref{GaY}. This combination is determined by the arbitrary complex parameter $\alpha$.

\begin{figure}[h!]
\centerline{
\includegraphics[height=65mm,width=65mm]{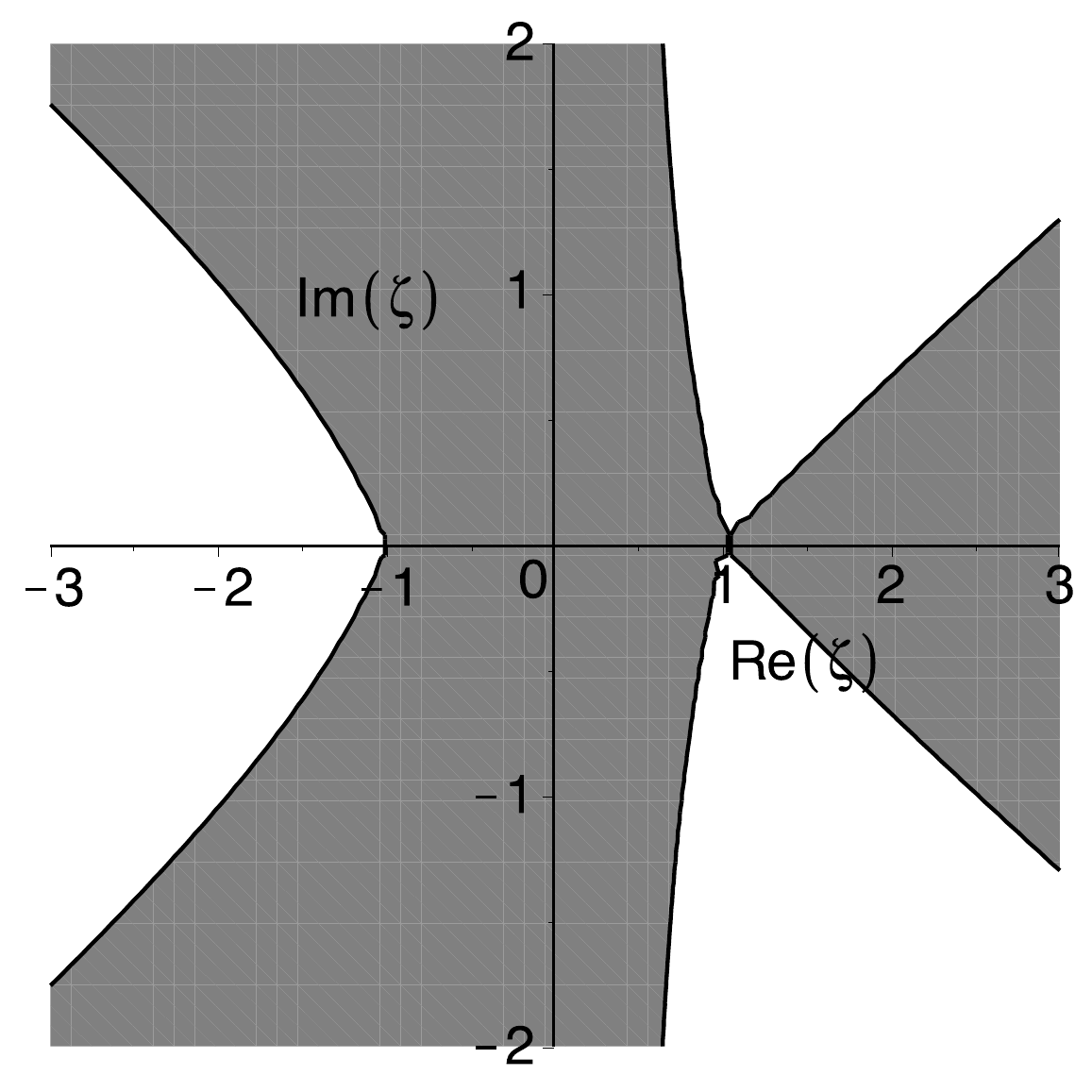}
\hspace*{10mm}
\includegraphics[height=63mm,width=72mm]{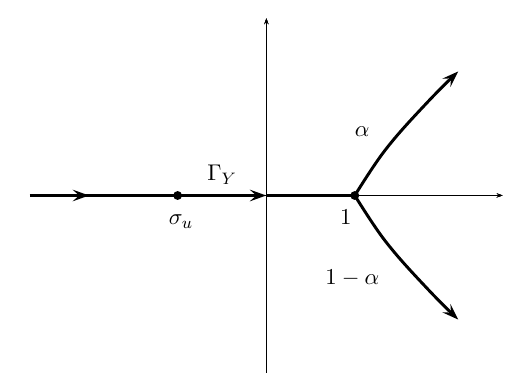}}
\vspace{5mm}
 \caption{On the left, level curves of
 $\textrm{Re}\,\phi_{cr}(\zeta)$ in the critical case $u=u_c$. In white, the region where $\textrm{Re}\, \phi_{cr}(\z)<0$, in grey, the region where $\textrm{Re}\, \phi_{cr}(\z)>0$. On the right, the contour $\Ga_Y$ for the Riemann-Hilbert formulation.}
\label{GaY}
\end{figure}

\subsection{Modified equilibrium measure near the critical case}
Following the ideas in \cite[\S 4.2]{DK}, but without the symmetry present in the quartic case, we construct a modified
equilibrium measure $\mu_u(x)$ in order to study the problem near
the critical case. This measure solves a minimization problem over
signed measures supported on the interval $[\sigma_u,1]$. The left
endpoint $\sigma_u$ will depend on $u$, and will become $-1$ when
$u=u_c$.

Because of the Euler--Lagrange equations for the equilibrium measure $\mu_u$, see for instance \cite{Deift}, we have
\begin{equation}\label{EL}
2\int \log \frac{1}{|x-y|}d\mu_u(y)+V(x;u)+\ell_u=0,
\qquad x\in(\sigma_u,1),
\end{equation}
for some constant $\ell_u$. Differentiating with respect to $x$, we
obtain
\begin{equation}\label{EL2}
2\cdot \operatorname{PV} \int \frac{d\mu_u(y)}{x-y}=V'(x;u),
\qquad x\in(\sigma_u,1).
\end{equation}

Differentiating \eqref{Vzeta}, we get
\begin{equation}\label{Vps}
V'(\zeta;u)=V'_{cr}(\zeta)+(u-u_c){V^o}'(\z)=-4\z^2+4\z+2-2\cdot3^{7/4}\left(2\zeta-1+\sqrt{3}\right)^2(u-u_c),
\end{equation}
using \eqref{zetaz} again. This modified equilibrium measure must verify
\begin{equation}\label{mem2a}
\begin{aligned}
\int_{\sigma_u}^1 d\mu_u(y)&=1,\\
\operatorname{PV} \int_{\sigma_u}^1 \frac{d\mu_u(y)}{x-y}&=-2x^2+2x+1-3^{7/4}\left(2x-1+\sqrt{3}\right)^2(u-u_c), \qquad x\in(\sigma_u,1).\\
\end{aligned}
\end{equation}

The resolvent,
\begin{equation}\label{mem1}
\om_u(\z)=\int_{\sigma_u}^1 \frac{d\mu_u(y)}{\z-y}, \qquad
\z\in \C\setminus[\sigma_u,1],
\end{equation}
because of \eqref{mem2a}, must satisfy the following:
\begin{equation}\label{mem2}
\begin{aligned}
&\om_u(\z)=\frac{1}{\z}+\mathcal{O}(\z^{-2}),\qquad \z\to\infty,\\
&\om_{u+}(x)+\om_{u-}(x)=-4x^2+4x+2-2\cdot 3^{7/4}\left(2x-1+\sqrt{3}\right)^2(u-u_c), \qquad x\in(\sigma_u,1).\\
\end{aligned}
\end{equation}
Consequently, we look for $\om_u(\z)$ in the form
\begin{equation}\label{mem3}
\om_u(\z)=-2\z^2+2\z+1-3^{7/4}\left(2\zeta-1+\sqrt{3}\right)^2(u-u_c)
-\frac{m_u(\zeta)}{2}\left(\frac{\z-\sigma_u}{\z-1}\right)^{1/2},
\qquad \z\in \C\setminus[\sigma_u,1],
\end{equation}
where the square root is taken on the principal
sheet, with a cut on $[\sigma_u,1]$. From the second equation in
\eqref{mem2}, it follows that $m_u(\z)$ has no jump on
$(\sigma_u,1)$, and hence it should be analytic in $\C$. From the
first equation in \eqref{mem2}, we obtain that $m_u(\z)$ is a
quadratic polynomial:
\begin{equation}\label{mem4}
m_u(\z)=a_2(\z-1)^2+a_1(\z-1)+a_0.
\end{equation}

Imposing the first condition in \eqref{mem2}, we can derive equations for the coefficients $a_0$, $a_1$ and $a_2$:
\begin{equation}\label{eqnsaj}
\begin{aligned}
a_2&=-4-8\cdot 3^{7/4}(u-u_c),\\
a_1&=-2(\sigma_u+1)-12\cdot 3^{3/4}(\sigma_u+2\sqrt{3}+1)(u-u_c)\\
a_0&=-\frac{1}{2}(\sigma_u+1)(3\sigma_u-5)
-3^{7/4}(3\sigma_u^2+(4\sqrt{3}-2)\sigma_u+7)(u-u_c).
\end{aligned}
\end{equation}

The endpoint $\sigma_u$ is a solution of the following cubic
equation:
\begin{equation}\label{cubicau}
\begin{aligned}
-\frac{1}{8}(\sigma_u+1)(5\sigma_u^2-14\sigma_u+13)
-\frac{3^{7/4}}{4}(\sigma_u-1)(5\sigma_u^2+(6\sqrt{3}-4)\sigma_u+7-2\sqrt{3})(u-u_c)=0.
\end{aligned}
\end{equation}

This last equation can be solved by perturbation, and we obtain
\begin{equation}\label{suseries}
\sigma_u=-1+3^{7/4}(2-\sqrt{3})(u-u_c)+\mathcal{O}((u-u_c)^2).
\end{equation}

Furthermore, the discriminant of the cubic equation \eqref{cubicau} is
\begin{equation}
\Delta_u=-16- 12\cdot 3^{1/4}(14\sqrt{3}+9)(u-u_c)+\mathcal{O}((u-u_c)^2),
\end{equation}
so for small $|u-u_c|$ we have $\Delta_u\neq 0$, and therefore we have three different solutions. In particular, the solution $\sigma_u$ is analytic in $u$ in a neighborhood of $u=u_c$.

Using \eqref{suseries}, we can write the coefficients $a_0$, $a_1$
and $a_2$ in powers of $u-u_c$:
\begin{equation}
\begin{aligned}
a_2&=-4-8\cdot 3^{7/4}(u-u_c),\\
a_1&=-3^{7/4}(4+6\sqrt{3})(u-u_c)+\mathcal{O}((u-u_c)^2),\\
a_0&=-12\cdot 3^{3/4}(u-u_c)+\mathcal{O}((u-u_c)^2).
\end{aligned}
\end{equation}

Higher order terms in $u-u_c$ can be computed by the same procedure. Note that when $u=u_c$ we have $m_{u_c}(\z)=-4(\z-1)^2$ and $\sigma_{u_c}=-1$, so
\begin{equation}
\om_{u_c}(\z)=-2\z^2+2\z+1+2(\z+1)^{1/2}(\z-1)^{3/2}
\qquad \z\in \C\setminus[-1,1],
\end{equation}
which recovers \eqref{omcr}. Finally, if we write the density as $d\mu_u(\z)=\psi_u(\z)d\z$, then
\begin{equation}\label{mem8}
\begin{aligned}
\psi_u(x)&=-\frac{1}{2\pi i}(\om_{u+}(x)-\om_{u-}(x))=
-\frac{m_u(x)\sqrt{x-\sigma_u}}{2\pi\sqrt{1-x}}, \qquad x\in (\sigma_u,1).
\end{aligned}
\end{equation}

\section{The Riemann--Hilbert problem}\label{RHP}


Following the work of Fokas, Its and Kitaev \cite{FIK}, consider the contour $\Ga_Y$ on Figure \ref{GaY} and
the following Riemann-Hilbert problem (RHP)
for a $2\times 2$ matrix-valued function $Y=Y_{N,n,u}\,:\; \C\setminus \Ga_Y\to \C^{2\times 2}$:
\begin{enumerate}
\item $Y(\z)$ is analytic in $\C\setminus \Ga_Y$, and for every $s\in\Ga_Y$ the limits
\begin{equation}\label{RHP_1}
Y_{\pm}(\z)=\lim_{s\to\z,\; s\in \Om_{\pm}}Y(s),
\end{equation}
exist, where $\Om_{\pm}$ are the left and the right sides of $\Ga_Y$, respectively, oriented as in Figure \ref{GaY}.
\item On $\Ga_Y$, the function $Y(\z)$ has a multiplicative jump:
\begin{equation}\label{RHP_2}
Y_{+}(\z)=Y_{-}(\z)
\begin{cases}
\begin{pmatrix}
1 & e^{-NV(\z;u)} \\
0 & 1
\end{pmatrix}, \qquad \z\in(-\infty,1]\\
\begin{pmatrix}
1 & \alpha e^{-NV(\z;u)} \\
0 & 1
\end{pmatrix}, \qquad \z\in\Ga_Y\cap\{\textrm{Im}\, \z>0\}\\
\begin{pmatrix}
1 & (1-\alpha)e^{-NV(s;u)} \\
0 & 1
\end{pmatrix}, \qquad \z\in\Ga_Y\cap\{\textrm{Im}\, \z<0\}.
\end{cases}
\end{equation}
\item As $\z\to\infty$,
\begin{equation}\label{RHP_3}
Y(\z)=\big(I+\mathcal O(\z^{-1})\big)
\begin{pmatrix}
\z^n & 0 \\
0 & \z^{-n}
\end{pmatrix}.
\end{equation}
\end{enumerate}

We call $n$ the {\it degree} of the RHP. This RHP has a unique
solution if and only if the monic polynomial $p_n(z)$, orthogonal
with respect to the weight function $w(z)$
uniquely exists. If additionally $p_{n-1}(z)$ uniquely exists, then the
solution of the RHP is given by:
\begin{equation}\label{Yz}
Y(\z)=Y_{n}(\z)= \begin{pmatrix} p_n(\z)  & (\mathcal{C}p_n w)(\z) \\[1mm]
    -\frac{2\pi i}{h_{n-1}}p_{n-1}(\z) & -\frac{2\pi i}{h_{n-1}}(\mathcal{C}p_{n-1}w)(\z)
    \end{pmatrix},
\end{equation}
where
\begin{equation}\label{Cauchy}
    (\mathcal{C}f)(\z)=\frac{1}{2\pi i}
    \int_{\Gamma}\frac{f(s)}{s-\z}\ d s
\end{equation}
is the Cauchy transform of $f$ on $\Gamma$, and the coefficient $h_{n-1}$
is defined as
\begin{equation}\label{hn}
h_{n-1}=\int_{\Gamma}p^2_{n-1}(s)w(s)d s.
\end{equation}

Conversely, as a consequence of the jump matrix and orthogonality, see
for instance \cite{Deift}, we have the following result:
\begin{proposition} \label{exi} Suppose that the RHP \eqref{RHP_1}-\eqref{RHP_3} has a solution $Y_n(z)$ for some $n$. Then
the orthogonal polynomial $p_n(\z)$ exists uniquely.
\end{proposition}

A very important identity satisfied by the sequence of orthogonal polynomials, if it
exists, is the following three term recurrence relation, see for instance \cite[\S 3.2]{Sze} or \cite[\S 1.4]{Chi}:
\begin{proposition} \label{ttrr}
Suppose that RHP \eqref{RHP_1}-\eqref{RHP_3} has solutions
$Y_{n-1}(\z)$, $Y_n(\z)$, and $Y_{n+1}(\z)$ for the degrees $n-1$, $n$,
and $n+1$, respectively. Then the orthogonal polynomials
$p_{n-1}(\z)$, $p_n(\z)$, and $p_{n+1}(\z)$, which uniquely exist by
Proposition \ref{exi}, satisfy the three term recurrence relation,
\begin{equation}\label{rr1}
\z p_n(\z)=p_{n+1}(\z)+\beta_{N,n}p_n(\z)+\gamma_{N,n}^2 p_{n-1}(\z).
\end{equation}
\end{proposition}

We note that all the matrices in the subsequent transformations will depend on $u$, $n$ and $N$, but we will omit this for brevity.

\subsection{The $g$-function}
The $g$-function associated with the modified equilibrium measure $d\mu_u(x)$ is
\begin{equation}\label{g1}
g_u(\z)=\int_{\sigma_u}^1 \log (\z-y) d\mu_u(y).
\end{equation}
It has the following properties:
\begin{enumerate}
\item $g_u(\z)$ is analytic for $\z\in\C\setminus (-\infty,1]$, and it has limiting values
$g_{u\pm}(x)$ as $\z\to x\pm i0$, $x\in (-\infty,1)$.
\item For $\z\in\C\setminus (-\infty,1]$,
\begin{equation}\label{g1a}
\frac{dg_u(\z)}{d\z}=\om_u(\z).
\end{equation}
\item By \eqref{EL}, $g_u(x)$ satisfies the Euler-Lagrange equation,
\begin{equation}\label{g2}
g_{u+}(x)+g_{u-}(x)-V(x;u)-\ell_u=0,\qquad x\in(\sigma_u,1),
\end{equation}
\item As $\z\to\infty$,
\begin{equation}\label{g3}
g_u(\z)=\log \z+\mathcal O(\z^{-1}).
\end{equation}
\item By \eqref{g1}, the difference of boundary values of $g(x)$ on the real axis is
\begin{equation}\label{g4}
G_u(x)\equiv g_{u+}(x)-g_{u-}(x)=
\begin{cases}
0, \qquad x\in[1,\infty)\\
2\pi i\displaystyle\int_x^1
d\mu_u(y),\qquad x\in(\sigma_u,1)\\
2\pi i, \qquad x\in (-\infty,\sigma_u].
\end{cases}
\end{equation}
\item Since the density of $d\mu_u(x)$, $\psi_u(x)$, is analytic on $(\sigma_u,1)$, the function $G_u(x)$ is analytic on $(\sigma_u,1)$ too,
and by \eqref{g4} and the Cauchy--Riemann equations,
\begin{equation}\label{g4aa}
\frac{G_u(x+iy)}{dy}\bigg|_{y=0}=2\pi \psi_u(x),\qquad \sigma_u<x<1.
\end{equation}
\end{enumerate}

In what follows, the function
\begin{equation}\label{g4b}
\phi_u(\z)=2g_{u}(\z)-V(\z;u)-\ell_u,\qquad \z\in\C\setminus(-\infty,1],
\end{equation}
will be important. By adding equations \eqref{g2} and \eqref{g4}, we obtain that
\begin{equation}\label{g5}
\phi_{u+}(x)=2g_{u+}(x)-V(x;u)-\ell_u=2\pi i\int_x^1 \psi_u(y)d(y),\qquad x\in(\sigma_u,1).
\end{equation}
and also $\phi_{u+}(x)=G_u(x)$ for $x\in(\sigma_u,1)$. By analytic continuation, we can extend this equation from $x\in
(\sigma_u,1)$ to $\z\in\C\setminus(-\infty,1]$. To this end we first
extend the density $\psi_u(x)$, for $x\in(\sigma_u,1)$, to a
function $r_u(\z)$ analytic on $\C\setminus[\sigma_u,1]$ such that

\begin{equation}\label{g6}
r_{u\pm}(x)=\pm\psi_u(x),\qquad x\in(\sigma_u,1).
\end{equation}

Equation \eqref{g5} can be now analytically extended to $\z\in\C\setminus(-\infty,1]$ as
\begin{equation}\label{g9a}
\phi_u(\z)=2g_{u}(\z)-V_u(\z)-\ell_u=2\pi i\int_{\z}^1 r_u(s)ds,\qquad \z\in\C\setminus(-\infty,1],
\end{equation}
where integration is taken over the segment joining $\z$ and $1$ in the complex plane.

Observe that similar to \eqref{g5},
\begin{equation}\label{g12}
\phi_{u-}(x)=2g_{u-}(x)-V_u(x)-\ell_u=-2\pi i\int_x^1 \psi_u(y)dy,\qquad x\in(\sigma_u,1),
\end{equation}
and $\phi_{u-}(x)=-G_u(x)$ for $x\in(\sigma_u,1)$. Also,
\begin{equation}\label{g13}
\phi_{u\pm}(x)=\pm 2\pi i+2\pi i\int_x^{\sigma_u} r_u(y)dy,\qquad
x\in(-\infty,\sigma_u].
\end{equation}
The function $G_u(x)$, defined in \eqref{g4}, is analytic on the interval $(\sigma_u,1)$. By
\eqref{g5} and \eqref{g12} it is analytically extended to the set $\C\setminus\big((-\infty,\sigma_u]\cup [1,\infty)\big)$,
and
\begin{equation}\label{g14}
\phi_{u}(\z)=\pm G_u(\z)\quad \textrm{for}\quad \pm\Im\z>0.
\end{equation}

\subsection{First transformation of the RHP}
Define $T(\z)$ as
\begin{equation}\label{ft4}
T(\z)=e^{-\frac{N \ell_u\sg_3}{2}}Y(\z)e^{-N\left[g_u(\z)-\frac{\ell_u}{2}\right]\sg_3},
\qquad \sg_3=\begin{pmatrix} 1 & 0\\ 0 & -1\end{pmatrix}.
\end{equation}
Then $T(\z)$ solves the following RHP:
\begin{enumerate}
\item $T(\z)$ is analytic in $\C\setminus \Ga_Y$, and for every $\z\in\Ga_Y$ the following limits exist:
\begin{equation}\label{ft5}
T_{\pm}(\z)=\lim_{s\to\z,\; s\in \Om_{\pm}}T(s).
\end{equation}
\item The jumps for the matrix $T$ are as follows:
\begin{equation}\label{ft7}
T_{+}(\z)=T_{-}(\z)
\begin{cases}
\begin{pmatrix}
e^{-NG_u(\z)} & 0 \\
0 & e^{NG_u(\z)}
\end{pmatrix},\qquad \z\in(\sigma_u,1)\\
\begin{pmatrix}
1 &\alpha e^{N\phi_u(\z)}\\
0 & 1
\end{pmatrix},\qquad \z\in\Gamma_Y\cap\{\textrm{Im}\, \z>0\}\\
\begin{pmatrix}
1 &(1-\alpha) e^{N\phi_u(\z)}\\
0 & 1
\end{pmatrix},\qquad \z\in\Gamma_Y\cap\{\textrm{Im}\, \z<0\}\\
\begin{pmatrix}
1 & e^{N\tilde\phi_u(\z)} \\
0 & 1
\end{pmatrix},\qquad \z\in(-\infty,\sigma_u],
\end{cases}
\end{equation}
where
\begin{equation}\label{ft11}
\tilde\phi_u(x)
=2\pi i\int_x^{\sigma_u}r_u(y)dy,\qquad x\in(-\infty,\sigma_u],
\end{equation}
using \eqref{g5}, \eqref{g12} and \eqref{g13}.

\item As $\z\to\infty$,
\begin{equation}\label{ft8}
T(\z)=I+\mathcal O\left(\z^{-1}\right).
\end{equation}
\end{enumerate}

\subsection{Second transformation of the RHP: Opening of lenses}
The jump matrix $J_T(x;u)$ can be factored as follows on the interval $(\sigma_u,1)$:
\begin{equation}\label{st1}
\begin{pmatrix}
e^{-NG_u(x)} & 1 \\
0 & e^{NG_u(x)}
\end{pmatrix}=\begin{pmatrix}
1 & 0 \\
e^{NG_u(x)} & 1
\end{pmatrix}\begin{pmatrix}
0 & 1 \\
-1 & 0
\end{pmatrix}\begin{pmatrix}
1 & 0 \\
e^{-NG_u(x)} & 1
\end{pmatrix}.
\end{equation}

\begin{figure}
\centerline{
\includegraphics{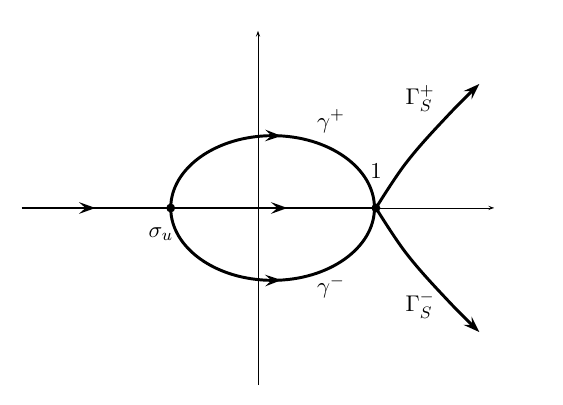}}
 \caption{The contour $\Ga_S$ for the Riemann-Hilbert analysis.}
\label{GaS}
\end{figure}

We introduce the contour $\Ga_S$ as shown on Figure \ref{GaS}, and we set
\begin{equation}\label{st2}
S(\z)=
\left\{
\begin{aligned}
&T(\z)
\begin{pmatrix}
1 & 0 \\
-e^{-NG_u(\z)} & 1
\end{pmatrix},\quad \textrm{ if $\z$ is in the upper part of the lens},\\
&T(\z)
\begin{pmatrix}
1 & 0 \\
e^{NG_u(\z)} & 1
\end{pmatrix},\quad \textrm{if $\z$ is in the lower part of the lens},\\
&T(\z),\quad {\rm otherwise.}
\end{aligned}
\right.
\end{equation}
Then $S(\z)$ solves the following RHP:
\begin{enumerate}
\item $S(\z)$ is analytic in $\C\setminus \Ga_S$, and for every $\z\in\Ga_S$ the following limits exist:
\begin{equation}\label{st3}
S_{\pm}(\z)=\lim_{s\to\z,\; s\in \Om_{\pm}}S(s).
\end{equation}
\item The matrix $S$ has the following jumps:
\begin{equation}\label{st5}
S_{+}(\z)=S_{-}(\z)
\left\{
\begin{aligned}
&\begin{pmatrix}
1 & e^{N\tilde\phi_u(\z)} \\
0 & 1
\end{pmatrix},\qquad \z\in(-\infty,\sigma_u],\\
&\begin{pmatrix}
0 & 1 \\
-1 & 0
\end{pmatrix},\qquad \z\in(\sigma_u,1),\\
&\begin{pmatrix}
1 & \alpha e^{N\phi_u(\z)} \\
0 & 1
\end{pmatrix},\qquad \z\in \Ga^+_S\\
&\begin{pmatrix}
1 & (1-\alpha)e^{N\phi_u(\z)} \\
0 & 1
\end{pmatrix},\qquad \z\in \Ga^-_S\\
&\begin{pmatrix}
1 & 0 \\
e^{-NG_u(\z)} & 1
\end{pmatrix},\qquad \z\in\ga^+,\\
&\begin{pmatrix}
1 & 0 \\
e^{NG_u(\z)} & 1
\end{pmatrix},\qquad \z\in\ga^-,
\end{aligned}
\right.
\end{equation}
where $\ga^+$ and $\ga^-$ are the curved boundaries of the upper and lower part of the lens, respectively. Recall that the function $\tilde{\phi}_u(s)$ is given by \eqref{ft11}.
\item As $\z\to\infty$,
\begin{equation}\label{st6}
S(\z)=I+\mathcal O\left(\z^{-1}\right).
\end{equation}
\end{enumerate}

Note that because of \eqref{g14}, we can rewrite the jump conditions on the lips of the lens as follows:
\begin{equation}
J_S(\z;u)=\begin{pmatrix}
1 & 0 \\
e^{-N\phi_u(\z)} & 1
\end{pmatrix},\qquad \z\in\ga^+\cup\ga^-.
\end{equation}

As observed in \cite{DK}, the use of the modified equilibrium measure causes some difficulties when opening the lens. If the density of the modified measure $\psi_u(y)>0$, then \eqref{g4} implies
\begin{equation}
\frac{d}{dx}\textrm{Im}\, G_u(x+0i)=-2\pi i \psi_u(x)<0
, \qquad x\in(\sigma_u,1),
\end{equation}
and then, because of the Cauchy--Riemann equations, we deduce that
\begin{equation}
\frac{d}{dy}\textrm{Re}\, G_u(x+iy)>0, \qquad x\in(\sigma_u,1),
\end{equation}
and the off--diagonal entries of the jump matrix $J_S(s;u)$ on the lips of the lens $\gamma_+\cup\gamma_-$, see \eqref{st5}, are exponentially decaying. However, the modified density $\psi_u(y)$ becomes  negative for $x$ close to $1$, so we obtain exponentially increasing terms. This can be controlled if $|u-u_c|$ is small enough, though, as shown in the following lemmas:

\begin{lemma} \label{ru}
Let $U$ be a neighborhood of the point $1$, given $\ep>0$ there exists $\delta>0$ such that for every $u\in\mathbb{R}$ with $|u-u_c|<\delta$, we have $\textrm{Re}\, \phi_u>\ep$ on the lips of the lens outside $U$.
\end{lemma}
\begin{proof}
The point where the density $\psi_u(x)$ changes sign in $[\sigma_u,1]$ can be computed explicitly to leading order in $u-u_c$, using \eqref{mem8}. In fact, we have two roots of the factor $m_u(\z)$:
\begin{equation}
x_{\pm}=1\pm 3^{7/8}\sqrt{5-2\sqrt{3}}\,(u_c-u)^{1/2}+\mathcal{O}(u-u_c),
\end{equation}
and so given $\ep>0$, we can find $\delta>0$ small enough such that for
$|u-u_c|<\delta$, the function $\psi_u$ is positive in the interval $(\sigma_u,x_-)$, and then $\textrm{Re}\,\phi_u>\ep$ on the lips of the lens, away from a small neighborhood of the endpoint $x=1$.
\end{proof}

On the contours $\Gamma^{\pm}_S$, we can prove the following:
\begin{lemma} \label{ru+} For any $\ep>0$ there exists $\de>0$ in such a way that if $|u-u_c|<\de$ then on the curve
\[
\Ga_S^{\pm\ep}=\Ga_S\cap \{\pm \Im \z\ge \ep\}
\]
there is a constant $C>0$ such that we have the inequality,
\begin{equation}\label{ft20}
\phi_u(\z)\le -C|\z-1|^3,\qquad s\in\Ga_Y^{\ep}.
\end{equation}
\end{lemma}

\begin{proof}
We use the fact that
\begin{equation}
\phi_u(\z)=\phi_{cr}(\z)+(u-u_c)\phi^o(\z),
\end{equation}
together with the explicit expression for $\phi_{cr}(\z)$ given in \eqref{phicr1exp}. In particular, we have
\begin{equation}
\phi_{cr}(\z)=\frac{4}{3}\z^3+\mathcal{O}(\z^2), \qquad\z\to\infty,
\end{equation}
so if we substitute $\z=1+re^{\pm\pi i/5}$, we have
\begin{equation}
\phi_{cr}(r)=Cr^3+\mathcal{O}(r^2), \qquad \textrm{Re}\, C<0.
\end{equation}

A similar estimate can be computed for the term $\phi^o(\z)$:
\begin{equation}
\phi^o(r)=\tilde{C}_u r^3+\mathcal{O}(r^2), \qquad \textrm{Re}\, C_u<0,
\end{equation}
if $|u-u_c|$ is small enough, and then we have the desired decay. 
\end{proof}

Lemma \ref{ru+} implies that for any $\ep>0$ there exists $\de>0$
such that if $|u-u_c|\le \de$ then
\begin{equation}\label{ft21}
J_S(\z;u)=
\begin{pmatrix}
1 & \mathcal O(e^{-c(\z)N}) \\
0 & 1
\end{pmatrix},\qquad \z\in\Ga_S^{\ep},
\end{equation}
where $c(\z)=C|\z-1|^3$, with $C>0$.

\subsection{Model solution}
The model solution $M(\z)$ solves the following  RHP:
\begin{enumerate}
\item $M(\z)$ is analytic in $\C\setminus [\sigma_u,1]$, and for every $x\in(\sigma_u,1)$ the following limits exist:
\begin{equation}\label{ms1}
M_{\pm}(\z)=\lim_{s\to \z,\; s\in \Om_{\pm}}M(s),
\end{equation}
\item On the interval $(\sigma_u,1)$,
\begin{equation}\label{ms2}
M_{+}(x)=M_{-}(x)\begin{pmatrix}
0 & 1 \\
-1 & 0
\end{pmatrix}.
\end{equation}
\item As $\z\to\infty$,
\begin{equation}\label{ms3}
M(\z)=I+\mathcal O\left(\z^{-1}\right).
\end{equation}
\end{enumerate}
The model solution is explicitly equal to
\begin{equation}\label{ms4}
M(\z)=\frac{1}{2}
\begin{pmatrix}
1 & 1\\ i & -i
\end{pmatrix}
\begin{pmatrix}
\beta^{-1}(\z) & 0\\ 0 & \beta(\z)
\end{pmatrix}
\begin{pmatrix}
1 & 1\\ i & -i
\end{pmatrix}^{-1}
=
\begin{pmatrix}
\frac{\be(\z)+\be^{-1}(\z)}{2} & \frac{\be(\z)-\be^{-1}(\z)}{-2i}  \\
\frac{\be(\z)-\be^{-1}(\z)}{2i} & \frac{\be(\z)+\be^{-1}(\z)}{2}
\end{pmatrix},
\end{equation}
where
\begin{equation}\label{ms5}
\be(\z)=\left(\frac{\z-\sigma_u}{\z-1}\right)^{1/4},
\end{equation}
with a cut on $[\sigma_u,1]$ and the branch such that $\be(\infty)=1$.

\subsection{Parametrix around the point $\z=1$}
The local parametrix near the endpoint $\z=1$ will be given in terms
of a certain solution of the Painlev\'e I differential equation. We
consider a disc $D(1,\ep)=\{\z:\; |\z-1|\le\ep\}$ around the right
endpoint of the support. We seek a function $P(\z)$ that satisfies
the following Riemann--Hilbert problem:
\begin{enumerate}
\item $P(\z)$ is analytic in $\mathbb{C}\setminus\Gamma_S$, see Figure \ref{GaS}, and for every $s\in\Ga_S$ the limits
\begin{equation}
P_{\pm}(\z)=\lim_{s\to \z,\; s\in \Om_{\pm}}P(s),
\end{equation}
exist.
\item On $D(1,\ep)\cap\Gamma_S$, the function $P$ has the following jumps:
\begin{equation}\label{P_jumps}
P_{+}(\z)=P_{-}(\z)
\begin{cases}
\begin{pmatrix} 1 & \alpha e^{N\phi_u(\z)}\\ 0 & 1
\end{pmatrix}, \qquad \z\in D(1,\ep)\cap \Gamma_S^+,\\
\begin{pmatrix} 1 & (1-\alpha)e^{N\phi_u(\z)}\\ 0 & 1
\end{pmatrix}, \qquad \z\in D(1,\ep)\cap \Gamma_S^-,\\
\begin{pmatrix} 1 & 0\\ e^{-N\phi_u(\z)} & 1
\end{pmatrix}, \qquad \z\in D(1,\ep)\cap\left(\gamma^+\cup\gamma^-\right),\\
\begin{pmatrix} 0 & 1\\ -1 & 0\end{pmatrix}, \qquad \z\in D(1,\ep)\cap(\sigma_u,1).
\end{cases}
\end{equation}
\item As $\z\to\infty$,
\begin{equation}\label{PM_matching}
P(\z)=M(\z)\left(I+\mathcal{O}(N^{-1/5})\right),
\end{equation}
uniformly for $\z\in\partial D(1,\ep)$.

\end{enumerate}

We observe that if we consider the function
\begin{equation}\label{px1}
\tilde{P}(\z)=P(\z)e^{N\phi_u(\z)\sg_3/2}.
\end{equation}
in $D(1,\ep)$, then by \eqref{st5}, this new function has the following jumps in $\Ga_{\tilde{P}}=\Ga_S\cap D(1,\ep)$:
\begin{equation}\label{px5}
\tilde{P}_{+}(\z)=\tilde{P}_{-}(\z)
\begin{cases}
\begin{pmatrix}
0 & 1 \\
-1 & 0
\end{pmatrix},\qquad \z\in [1-\ep,1),\\
\begin{pmatrix}
1 & \alpha \\
0 & 1
\end{pmatrix},\qquad \z\in \Ga_S^+\cap D(1,\ep),\\
\begin{pmatrix}
1 & 1-\alpha \\
0 & 1
\end{pmatrix},\qquad \z\in \Ga_S^-\cap D(1,\ep),\\
\begin{pmatrix}
1 & 0 \\
1 & 1
\end{pmatrix},\qquad \z\in \ga^{\pm}.
\end{cases}
\end{equation}

To see this, we recall that $\phi_{u\pm}(x)=\pm G_u(x)$ for $x\in(\sigma_u,1)$ by formula \eqref{g5} and \eqref{g12}, and also $\phi_{u}(\z)=\pm G_u(\z)$ for $\z\in\ga^{\pm}$ because of \eqref{g14}.

The jumps $J_{\tilde{P}}(\z;u)$ fit well to the ones of the $\Psi(w;\lambda,\alpha)$-functions given in Section \ref{RHforPsi}. In the spirit of \cite{DK}, we look for a local parametrix in a neighborhood of the endpoint $\z=1$ using this function:
\begin{lemma} Let $P:D(1,\varepsilon)\setminus\Sigma_S\mapsto \mathbb{C}^{2\times 2}$ be defined by the formula
\begin{equation}\label{px9}
P(\z)=E(\z)\Psi(N^{2/5}f(\z);N^{4/5}h_u(\z),\alpha)e^{-N\phi_u(\z)\sg_3/2},
\end{equation}
where
\begin{itemize}
\item the function $\Psi(w;\lambda,\alpha)$ is the solution of the Riemann--Hilbert problem stated in Section \ref{RHforPsi}.
\item $f(\z)$ is a conformal map from $D(1,\ep)$ to a neighborhood of $0$ such that $\Gamma_S\cap D(1,\ep)$ (see Figure \ref{fig_mapping})
is mapped onto part of the contour in Figure \ref{PI_jumps}. Explicitly,
\begin{equation}\label{fzeta}
f(\z)=\left[\frac{5}{8}\phi_{cr}(\z)\right]^{2/5}.
\end{equation}
\item $h_u:D(1,\ep)\mapsto \mathbb{C}$ is analytic and such that $N^{4/5}h_u(D(1,\ep))$ does not contain any poles of the solution $y_{\alpha}(\lambda)$ of Painlev\'e I. Explicitly,
\begin{equation}\label{hu}
h_u(\z)=\left(\frac{1}{20}\right)^{1/5}\,\frac{\phi_u(\z)-\phi_{cr}(\z)}{\phi_{cr}(\z)^{1/5}}.
\end{equation}
\item $E:D(1,\ep)\mapsto \mathbb{C}^{2\times 2}$ is a suitable analytic prefactor:
\begin{equation}\label{Ez}
E(\z)=M(\z)\left[\frac{(N^{2/5}f(\z))^{\sg_3/4}}{\sqrt{2}}\begin{pmatrix} 1 & -i\\
1 & i \end{pmatrix}\right]^{-1} =\frac{1}{\sqrt{2}} M(\z)
\begin{pmatrix} 1 & 1\\ i & -i \end{pmatrix}
(N^{2/5}f(\z))^{-\sg_3/4},
\end{equation}
\end{itemize}
then $P(\z)$ is analytic in $D(1,\varepsilon)\setminus\Sigma_S$
and it satisfies the Riemann--Hilbert problem stated before.
\end{lemma}
\begin{proof}
The proof is based on the matrix $\tilde{P}(\z)$ presented before,
in particular note the jumps \eqref{px5}.
\end{proof}

\begin{figure}
\centerline{
\includegraphics{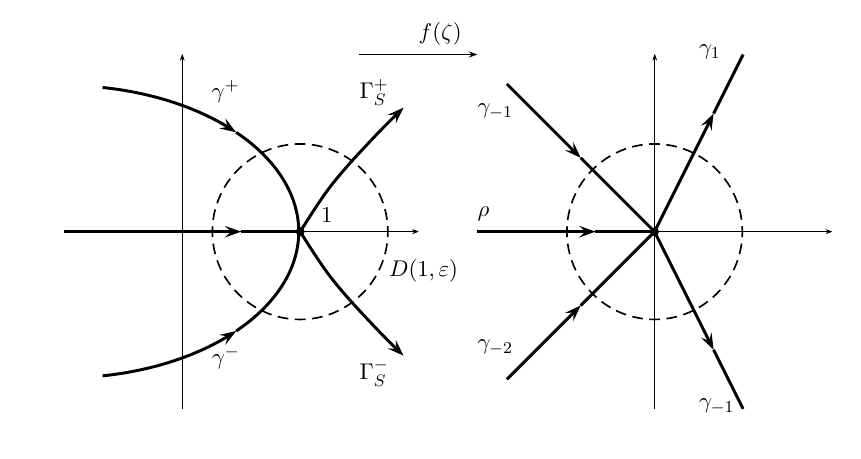}}
 \caption{Construction of the local parametrix in a neighborhood of $\z=1$.}
\label{fig_mapping}
\end{figure}

We have the following result:
\begin{lemma} Suppose that $\lambda\in\mathbb{R}$ is not a pole of $y_{\alpha}(\lambda)$, and $u$ depends on $N$ and $\lambda$ as follows:
\begin{equation}
N^{4/5}(u-u_c)=c_1 \lambda, \qquad c_1=2^{-12/5}3^{-7/4}.
\end{equation}

Let $f(\zeta)$, $h_u(\zeta)$ and $E(\zeta)$ be as defined in
\eqref{fzeta}, \eqref{hu} and \eqref{Ez}, then there is a disc around the point
$\z=1$, which depends only on $\lambda$, such that the previous properties
and the matching condition \eqref{PM_matching} are satisfied.
\end{lemma}

\begin{proof}

We recall that
\begin{equation}
\phi_{cr}(\z)=4\int_1^{\z} (s+1)^{1/2}(s-1)^{3/2}ds,
\end{equation}
and we can write
\begin{equation}\label{phicr_at_1}
\phi_{cr}(\z)=\frac{8\sqrt{2}}{5}(\z-1)^{5/2}\left(1+v_{cr}(\z)\right),
\end{equation}
where the function $v_{cr}(\z)$ is analytic in a neighborhood of $\z=1$, with $v_{cr}(1)=0$. Therefore, if we take a small enough disc around $\z=1$, then $f(\z)$ given by \eqref{fzeta} is a conformal map to a neighborhood of the origin. Actually,
\begin{equation}\label{fz1}
f(\z)=2^{1/5}(\z-1)+\mathcal{O}((\z-1)^2), \qquad \z\to 1,
\end{equation}
so in particular $f(\z)$ maps the interval $(1-\ep,1)$ onto a
part of the negative real axis.

The curve $\Gamma_S^+$ is that of steepest descent for
$\phi_{cr}(\z)$, so on $\Gamma_S^+$ the function $\phi_{cr}(\z)$ is
real and negative, i.e. $\arg \phi_{cr}(\z)=\pi$. Therefore, $\arg
f(\z)=2\pi/5$ when $\z\in\Gamma_S^+$.

Finally, the location of the lips of the lens inside the disc $D(1,\varepsilon)$ is chosen in such a way  that they are mapped by $f(\z)$ onto part of the rays $\arg\z=\pm 4\pi/5$ (upper and lower lip respectively), taking into account \eqref{fzeta}. This is achieved by requiring that in a neighborhood of $\z=1$, we have $\arg\phi_{cr}(\z)=\pm 2\pi$ on the upper and lower lips of the lens (respectively).

Regarding the properties of the map $h_u(\z)$, we can write
\begin{equation}
\phi_u(\z)-\phi_{cr}(\z)=(u-u_c)\phi^o(\z,u).
\end{equation}

The function $\phi^o(\z,u)$ satisfies
\begin{equation}
\phi^o(\z,u)= 2^{7/2} 3^{7/4}(\z-1)^{1/2}(1+v^o(\z,u)),
\end{equation}
where $v^o(\z,u)$ is analytic in $\z$ in a neighborhood of $\z=1$ and $v^o(1,u)=0$. This follows from formula \eqref{g9a}, namely
\begin{equation}
\phi_u(\z)=2\pi i\int_{\z}^1 r_u(s)ds,
\end{equation}
where $r_u(s)$ is the analytic extension of $\varrho_u(s)$ to $\mathbb{C}\setminus(-\infty,1]$, so
\begin{equation}
r_u(s)=\frac{m_u(s)\sqrt{s-\sigma_u}}{2\pi\sqrt{s-1}}=\frac{M_u(s)}{\sqrt{s-1}}, 
\end{equation}
with $M_u(s)$ analytic near $s=1$. Also, note that from \eqref{phicr_at_1} we have
\begin{equation}
\phi_{cr}(\z)^{1/5}=2^{7/10}5^{-1/5}(\z-1)^{1/2}\left(1+\tilde{v}_{cr}(\z)\right),
\end{equation}
with $\tilde{v}_{cr}(\z)$ analytic in a neighborhood of $\z=1$ and
$\tilde{v}_{cr}(1)=0$. We conclude that
\begin{equation}
h^o(\z):=\left(\frac{1}{20}\right)^{1/5}\,\frac{\phi^o(\z)}{\phi_{cr}(\z)^{1/5}}
=2^{12/5}3^{7/4}\left(1+\mathcal{O}(\z-1)\right),
\end{equation}
is analytic in the variable $\z$ in a neighborhood of $\z=1$. Observe that, differently from \cite[Lemma 4.3]{DK}, the error term in this case depends on $u$ (actually on powers of $u-u_c$), but that does not change the properties of $h_u(\z)$ as a function of $\z$ locally near $\z=1$.

As a consequence of the last formula, we can define
$c_1^{-1}=h^o(1)=2^{12/5}3^{7/4}$, and then
\begin{equation}
N^{4/5}h_u(1)=N^{4/5}(u-u_c)h^o(1)=N^{4/5}(u-u_c)c_1^{-1}=\lambda.
\end{equation}
We suppose that $\lambda$ is not a pole of $y_{\alpha}(\lambda)$, so we can find a
neighborhood $D(1,\varepsilon)$ of $\z=1$, for $\varepsilon$ small
enough, such that $N^{4/5}h_u(\z)$ is not a pole of $y_{\alpha}(\lambda)$
for all $\z\in D(1,\varepsilon)$.

Finally, substituting formulas \eqref{ms4} and \eqref{ms5} for $M(\z)$, we have
\begin{equation}
E(\z)=\frac{1}{\sqrt{2}}\begin{pmatrix}1 & 1\\ i &-i \end{pmatrix}
\left(N^{2/5}f(\z)\,\frac{\z-\sigma_u}{\z-1}\right)^{-\sigma_3/4},
\end{equation}
so it is an analytic function in a neighborhood of $\z=1$, because $f(\z)$ has a simple zero at $\z=1$, see \eqref{fz1}. Moreover, using the asymptotic behavior of the $\Psi(\z;\lambda,\alpha)$ function, see \eqref{rhp8A}, we have
\begin{equation}
P(\z)=E(\z)\frac{(N^{2/5}f(\z))^{\sigma_3/4}}{\sqrt{2}}\begin{pmatrix}1
& -i\\ 1 & i\end{pmatrix}
\left(I+\mathcal{O}(N^{-1/5})\right)e^{\theta(\z;\lambda)\sg_3-N\phi_u(\z)\sigma_3/2} 
\end{equation}
as $N\to\infty$. Substituting the expression for $\theta(\z;\lambda)$ and the formulas for $f(\z)$ and $h_u(\z)$, see \eqref{fzeta} and \eqref{hu}, we obtain
\begin{equation}
\theta(\z;\lambda)=\theta(N^{2/5}f(\z);N^{4/5}h_u(\z))
=\frac{4}{5}(N^{2/5}f(\z))^{5/2}+N^{4/5}h_u(\z)(N^{2/5}f(\z))^{1/2}
=\frac{N\phi_u(\z)}{2},
\end{equation}
so the exponential factor cancels out. Hence, substituting the expression for $E(\z)$, we get 
\begin{equation}
P(\z) = \frac{1}{2}\begin{pmatrix}1 & 1\\ i &-i
\end{pmatrix}\left(\frac{\z-1}{\z+1}\right)^{\sigma_3/4}
\begin{pmatrix}1 & -i\\ 1 & i\end{pmatrix}\left(I+\mathcal{O}(N^{-1/5})\right)
= M(\z)\left(I+\mathcal{O}(N^{-1/5})\right),
\end{equation}
uniformly for $\z\in\partial D(1,\varepsilon)$.
\end{proof}

\subsection{Parametrix near the endpoint $\z=\sigma_u$}

Near the endpoint $\z=\sigma_u$, we take a small disc $D(\sigma_u,\varepsilon)=\{\z:|\z-\sigma_u|<\varepsilon\}$, and the local parametrix should satisfy the following RH problem:
\begin{enumerate}
\item $Q(\z)$ is analytic in $\mathbb{C}\setminus\Gamma_S$, see Figure \ref{GaS}.
\item On $\Gamma_S$ we have
\begin{equation}\label{Q_jumps}
Q_{+}(\z)=Q_{-}(\z)
\begin{cases}
\begin{pmatrix} 1 & e^{N\tilde{\phi}_u(\z)}\\ 0 & 1
\end{pmatrix}, \qquad \z\in D(\sigma_u,\ep)\cap(-\infty,\sigma_u],\\
\begin{pmatrix} 1 & 0\\ e^{-N\phi_u(\z)} & 1
\end{pmatrix}, \qquad \z\in D(\sigma_u,\ep)\cap\left(\gamma^+\cup\gamma^-\right),\\
\begin{pmatrix} 0 & 1\\ -1& 0\end{pmatrix}, \qquad \z\in D(\sigma_u,\ep)\cap(\sigma_u,1).
\end{cases}
\end{equation}
Here the function $\tilde{\phi}_u(\z)$ is defined by \eqref{ft11}.

\item As $N\to\infty$, uniformly for $\z\in\partial D(\sigma_u,\ep)$, we have
$$
Q(\z)=M(\z)\left(I+\mathcal{O}(N^{-1})\right).
$$
\end{enumerate}

Note that if we write
\begin{equation}
\tilde{Q}(\z)=Q(\z)e^{N\tilde{\phi}_u(\z)\sigma_3/2},
\end{equation}
then the matrix $\tilde{Q}(\z)$ satisfies the following Riemann--Hilbert problem:
\begin{enumerate}
\item $\tilde{Q}(\z)$ is analytic in $D(\sigma_u,\varepsilon)\setminus \Gamma_S$.
\item $\tilde{Q}(\z)$ has the following jumps:
\begin{equation}\label{Qt_jumps}
J_{\tilde{Q}}(\z;u)=
\begin{cases}
\begin{pmatrix} 1 & 1\\ 0 & 1
\end{pmatrix}, \qquad \z\in D(\sigma_u,\ep)\cap(-\infty,\sigma_u],\\
\begin{pmatrix} 1 & 0\\ 1 & 1
\end{pmatrix}, \qquad \z\in D(\sigma_u,\ep)\cap\left(\gamma^+\cup\gamma^-\right),\\
\begin{pmatrix} 0 & 1\\ -1 & 0\end{pmatrix}, \qquad \z\in D(\sigma_u,\ep)\cap(\sigma_u,1).
\end{cases}
\end{equation}
\item Uniformly for $\z\in\partial D(\sigma_u,\ep)$, we have
\begin{equation}\label{matchingQu}
\tilde{Q}(\z)=M(\z)\left(I+\mathcal{O}(N^{-1})\right)e^{N\tilde{\phi}_u(\z)\sigma_3/2}, \qquad N\to\infty,
\end{equation}
\end{enumerate}

\begin{figure}
\centerline{
\includegraphics{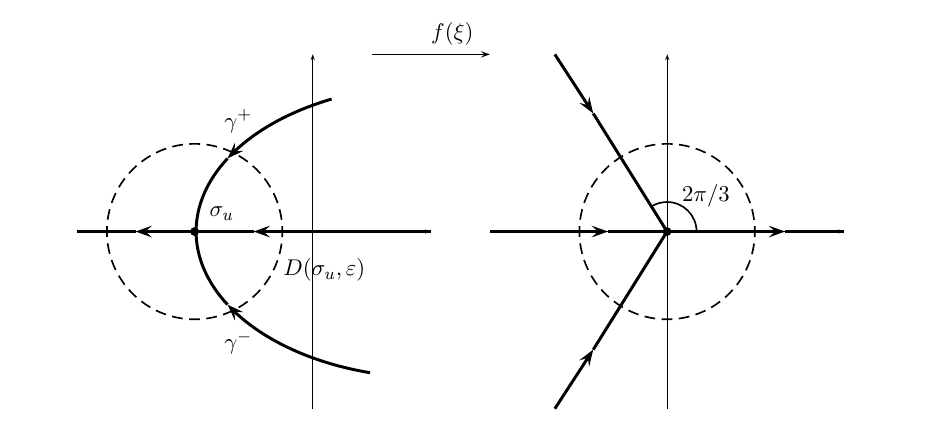}}
 \caption{Construction of the local parametrix in a neighborhood of $\z=\sigma_u$.}
\label{fig_mapping_Airy}
\end{figure}

This problem can be solved in terms of Airy functions. We write $\xi=-\z$, and then
\begin{equation}
\tilde{Q}(\z)=\sigma_3 \tilde{Q}(\xi)\sigma_3,
\end{equation}
where $\tilde{Q}(\xi)$ solves a standard Airy Riemann--Hilbert problem near the endpoint $-\sigma_u>0$.

In order to construct the matrix the local parametrix, we consider the functions
\begin{equation}
y_0(\xi)=\textrm{Ai}\,(\xi), \qquad
y_1(\xi)=\omega\textrm{Ai}\,(\omega\xi), \qquad
y_2(\xi)=\omega^2\textrm{Ai}\,(\omega^2\xi),
\end{equation}
where $\omega=e^{2\pi i/3}$, which are solutions of the Airy differential equation $y''(\xi)=\xi y(\xi)$. We construct the following function:
\begin{equation}
A(\xi)=
\begin{cases}
\begin{pmatrix}
y_0(\xi) & -y_2(\xi)\\
-iy'_0(\xi) & iy'_2(\xi)
\end{pmatrix}, \qquad \textrm{arg}\,\xi\in\left(0,\frac{2\pi}{3}\right),\\
\begin{pmatrix}
-y_1(\xi) & -y_2(\xi)\\
iy'_1(\xi) & iy'_2(\xi)
\end{pmatrix}, \qquad \textrm{arg}\,\xi\in\left(\frac{2\pi}{3},\pi\right),\\
\begin{pmatrix}
-y_2(\xi) & y_1(\xi)\\
iy'_2(\xi) & -iy'_1(\xi)
\end{pmatrix}, \qquad \textrm{arg}\,\xi\in\left(-\pi,-\frac{2\pi}{3}\right),\\
\begin{pmatrix}
y_0(\xi) & y_1(\xi)\\
-iy'_0(\xi) & -iy'_1(\xi)
\end{pmatrix}, \qquad \textrm{arg}\,\xi\in\left(-\frac{2\pi}{3},0\right).
\end{cases}
\end{equation}

This matrix--valued function satisfies the jumps in \eqref{Qt_jumps}, on contours emanating from the origin with angles $\pm 2\pi/3$ and $\pi$, as a consequence of the relation $y_0(\xi)+y_1(\xi)+y_2(\xi)=0$.

We also note the asymptotic behavior of the Airy function and its derivative, as the argument grows large:
\begin{equation}\label{asympAiry}
A(w)=w^{-\sigma_3/4}\frac{1}{2\sqrt{\pi}}\begin{pmatrix}1& i\\ i & 1\end{pmatrix}\left(I+\mathcal{O}(w^{-3/2})\right)e^{-\frac{2}{3}w^{3/2}\sigma_3}, \qquad w\to\infty,
\end{equation}
which follows from the standard asymptotic expansion of the Airy function and its derivative, see for instance
\cite[9.7.5, 9.7.6]{DLMF}.

Now we consider the following local parametrix:
\begin{equation}
\tilde{Q}(\xi)= E(\xi)A(N^{2/3}\tilde{f}(\xi)),
\end{equation}
where $E(\xi)$ is an analytic prefactor, $A(w)$ solves the
standard Airy Riemann--Hilbert problem and $w=N^{2/3}\tilde{f}(\xi)$ is a
conformal map from $D(-\sigma_u,\ep)$ to a neighborhood of the
origin in the auxiliary $w$ plane. 

The function $w=N^{2/3}\tilde{f}(\xi)$ is determined in such a way that we obtain a matching between the exponential factors in \eqref{matchingQu} and \eqref{asympAiry}:
\begin{equation}\label{fxi}
-\frac{2}{3}\left[N^{2/3}\tilde{f}(\xi)\right]^{3/2}=\frac{1}{2}N\tilde{\phi}_u(\xi)
\Rightarrow
\tilde{f}(\xi)=\left(-\frac{3}{4}\tilde{\phi}_u(\xi)\right)^{2/3}.
\end{equation}

This is a conformal mapping in a neighborhood of the point $-\sigma_u$ for $\ep$ and $u-u_c$ small enough, because of the properties of $\tilde{\phi}_u$.

Regarding the factor $E(\xi)$, we construct it in such a way that it takes care of the matching between the local and global parametrices:
\begin{equation}
\begin{aligned}
E(\xi)&=\sqrt{\pi}\, M(\xi)\begin{pmatrix}1& -i\\ -i & 1\end{pmatrix}
\left(N^{2/3}\tilde{f}(\xi)\right)^{\sigma_3/4}.
\end{aligned}
\end{equation}

The fractional power $\tilde{f}(\xi)^{\sigma_3/4}$ has a jump on the interval $[-\sigma_u-\delta,-\sigma_u]$, that is compensated with the jump of $M(\xi)$, and it is not difficult to check that $E(\xi)$ is indeed analytic in a neighborhood of $\xi=-\sigma_u$. Thus, the local parametrix is
\begin{equation}
\begin{aligned}
 Q(\z)&=\tilde{Q}(\z)e^{-N\tilde{\phi}_u(\z)\sigma_3/2}\\
&=\sigma_3 \tilde{Q}(\xi)\sigma_3 e^{-N\tilde{\phi}_u(\z)\sigma_3/2}\\
&=\sigma_3 E(\xi)A(N^{2/3}\tilde{f}(\xi))\sigma_3 e^{-N\tilde{\phi}_u(\z)\sigma_3/2}\\
\end{aligned}
\end{equation}

We also note that the matrix $A(N^{2/3}\tilde{f}(\xi))$ admits a full asymptotic
expansion in powers of $N^{-1}$, because of \eqref{asympAiry} and so does $\tilde{Q}(\z)$, uniformly for $\z$ in a neighborhood of $\z=\sigma_u$.

\subsection{Third transformation of the RHP}
The final step of the Riemann--Hilbert analysis is to define a matrix $R_u(\z)$ as follows:
\begin{equation}\label{Ru1}
R(\z)=S(\z)
\left\{
\begin{aligned}
&M^{-1}(\z), \qquad \z\in\mathbb{C}\setminus(\overline{D(1,\ep)}\cup\overline{D(\sigma_u,\ep)}\cup\Ga_S),\\
&P^{-1}(\z), \qquad \z\in D(1,\ep)\setminus \Ga_S,\\
&Q^{-1}(\z), \qquad \z\in D(\sigma_u,\ep)\setminus \Ga_S,\\
\end{aligned}
\right.
\end{equation}
where $\Ga_S$ is the contour shown in Figure \ref{GaS}. This matrix $R(\z)$ satisfies the following Riemann--Hilbert problem:
\begin{enumerate}
\item $R(\z)$ is analytic in $\C\setminus \Ga_R$, where $\Ga_R$ is shown in Figure \ref{GaR}.
\item On $\Ga_R$,
\begin{equation}\label{jumpsR}
R_{+}(\z)=R_{-}(\z)
\left\{
\begin{aligned}
&M(\z)\begin{pmatrix}
1 & e^{N\tilde\phi_u(\z)} \\
0 & 1
\end{pmatrix}M^{-1}(\z),\qquad \z\in\Ga_{R,1},\\
&M(\z)\begin{pmatrix}
1 & 0 \\
e^{-N\phi_u(\z)} & 1
\end{pmatrix}M^{-1}(s),\qquad \z\in \Ga_{R,2}\cup\Ga_{R,4},\\
&M(\z)\begin{pmatrix}
1 & \alpha e^{N\phi_u(\z)} \\
0 & 1
\end{pmatrix}M^{-1}(\z),\qquad \z\in\Ga_{R,3},\\
&M(\z)\begin{pmatrix}
1 & (1-\alpha)e^{N\phi_u(\z)} \\
0 & 1
\end{pmatrix}M^{-1}(\z),\qquad \z\in\Ga_{R,5},\\
&P(\z)M^{-1}(\z),\qquad \z\in \partial D(1,\ep),\\
&Q(\z)M^{-1}(\z),\qquad \z\in \partial D(\sigma_u,\ep),\\
\end{aligned}
\right.
\end{equation}
\item The matrix $R(\z)$ has the following asymptotic behavior:
\begin{equation}\label{asympR}
R(\z)=I+\mathcal O\left(\z^{-1}\right),\qquad \z\to\infty.
\end{equation}
\end{enumerate}

\begin{figure}[h]
\centerline{\includegraphics{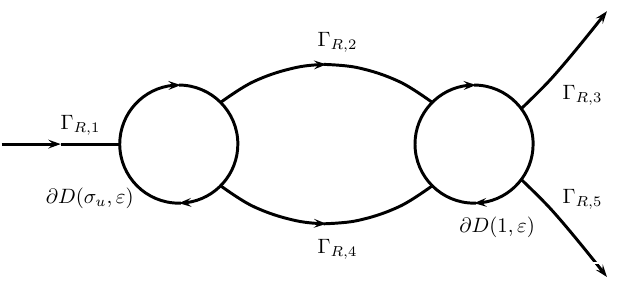}}
  \caption[The contour of integration ]{The contour of integration $\Ga_R$.}
\label{GaR}
 \end{figure}

From the matching conditions with $M(\z)$ and the last two
equations in \eqref{jumpsR}, we find that the jump matrix $J_R(s;u)$
can be expanded in inverse powers of $N^{-1/5}$ and $N^{-1}$ on the
boundary of the discs $D(1,\ep)$ and $D(\sigma_u,\ep)$ respectively.
Namely, as $N\to\infty$,
\begin{equation}\label{jumpsRcircles}
J_R(\z;u)\sim
\begin{cases}
I+\displaystyle\sum_{k=1}^{\infty} W^{(k)}(\z;u) N^{-k/5},\qquad
\z\in\partial D(1,\ep)\\
I+\displaystyle\sum_{k=1}^{\infty} \widehat{W}^{(k)}(\z;u)
N^{-k},\qquad \z\in\partial D(\sigma_u,\ep)
\end{cases}
\end{equation}

The terms in these expansions can be written as follows:
\begin{equation}
W^{(k)}(\z;u)=M(\z)\Psi_k(N^{4/5}h_u(\z),\alpha)f(\z)^{-k/2}M^{-1}(\z) ,
\qquad k\geq 1
\end{equation}
and
\begin{equation}
\widehat{W}^{(k)}(\z;u)=M(\z)\sigma_3
A_k\tilde{f}(-\z)^{-3k/2}\sigma_3 M^{-1}(\z), \qquad k\geq 1,
\end{equation}
with the coefficients appearing in \eqref{rhp8A} and in the asymptotic expansion of the Airy matrix in \eqref{asympAiry}, respectively.

As a consequence, the matrix $R(\z)$ itself admits an asymptotic
expansion in powers of $N^{-1/5}$ that holds uniformly for
$\z\in\mathbb{C}\setminus\Gamma_R$, where the contour $\Gamma_R$ is
depicted in Figure \ref{GaR}.

\begin{proposition}\label{propRn} The Riemann--Hilbert problem for $R(\z)$ has a unique
solution for large enough $N$, and the matrix $R(\z)$ can be
expanded in asymptotic series
\begin{equation}
R(\z)\sim I+\sum_{k=1}^{\infty} R^{(k)}(\z)
\end{equation}
where $R^{(k)}(\z)=\mathcal{O}(N^{-k/5})$ as $N\to\infty$, uniformly for $\z\in\mathbb{C}\setminus\Gamma_R$.
\end{proposition}
\begin{proof} The proof follows the lines of \cite[\S 11]{Ble}, so we omit the details. Note that because of Lemma \ref{ru} and Lemma \ref{ru+}, the jump matrices on the contours $\Gamma_{R,j}$, $j=1,\ldots, 5$ are exponentially close to the identity. On the circles around the endpoints we have expansions in inverse powers of $N$ and $N^{1/5}$, see \eqref{jumpsRcircles}. Thus the leading contribution is given by the expansion on $\partial D(1,\ep)$. Alternatively, one can derive RH problems for higher order terms in the expansion of $R(\z)$, similarly to \cite{DK}.

\end{proof}

\section{Proof of Theorem \ref{Th1}}

As a consequence of the previous steepest descent analysis of the
Riemann--Hilbert problem, we have proved the following:
\begin{itemize}
\item The Riemann--Hilbert problem is solvable for large enough values of $N$, and the solution $Y(\z)$ is unique. This follows from the solvability of the Riemann--Hilbert for $R(\z)$ and the fact that all the transformations $Y\mapsto T\mapsto S\mapsto R$ are invertible.
\item In particular, the $(1,1)$ entry of $Y$ exists for large enough $N$, and it is the $n$-th orthogonal polynomial $p_n(\z)$ with respect to the weight function $e^{-NV(\z;u)}$, where $V(\z;u)$ is described in \eqref{Vzeta}, and orthogonality is defined  on the curve $\Gamma$ described in \eqref{Gammal}. 
\end{itemize}

Moreover, if $P_{n}(\z,u)$ and also $P_{n\pm 1}(\z,u)$ exist, then they satisfy a three term recurrence relation:
\begin{equation}
\z P_{n}(\z,u)=P_{n+1}(\z,u)+\beta_{n,N}(u) p_n(\z,u)+\gamma_{n,N}^2(u) p_{n-1}(\z,u).
\end{equation}

In particular, this result is true in the diagonal case $n=N$, since we have a solution of the RH problem for large $N$. Then it is well known that the recurrence coefficients $\gamma^2_{N,N}(u)$ and $\beta_{N,N}(u)$ can be expressed in terms of the entries of the matrices in the Riemann--Hilbert analysis (for convenience, we omit the dependence on $u$ in the matrices):
\begin{equation}
\begin{aligned}
\gamma^2_{N,N}(u)&=[Y_1]_{12}[Y_1]_{21}=[T_1]_{12}[T_1]_{21},\\
\beta_{N,N}(u)&=\frac{[Y_2]_{12}}{[Y_1]_{12}}-[Y_1]_{22}=
\frac{[T_2]_{12}}{[T_1]_{12}}-[T_1]_{22},
\end{aligned}
\end{equation}
where
\begin{equation}
Y(\z)\z^{-N\sigma_3}=I+\frac{Y_1}{\z}+\frac{Y_2}{\z^2}+\ldots, \qquad
T(\z)=I+\frac{T_1}{\z}+\frac{T_2}{\z^2}+\ldots,
\qquad \z\to\infty.
\end{equation}

Now we recall that $T=S=RM$ outside the discs $D(1,\ep)$ and $D(\sigma_u,\ep)$ and the lens that we opened in the $T\mapsto S$ transformation, and we use the large $\z$ asymptotic expansion
\begin{equation}
\begin{aligned}
M(\z)&=I+\frac{M_1}{\z}+\frac{M_2}{\z^2}+\ldots\\
&=I+\begin{pmatrix}0 & i\\
-i & 0\end{pmatrix}\frac{1-\sigma_u}{4\z}
+\begin{pmatrix} 1-\sigma_u & 4i (1+\sigma_u)\\
-4i(1+\sigma_u)& 1-\sigma_u \end{pmatrix}\frac{1-\sigma_u}{32\z^2}+\ldots, \qquad \z\to\infty.
\end{aligned}
\end{equation}

Since
\begin{equation}
\begin{aligned}
S_1&=T_1=M_1+R_1,\\
S_2&=T_2=M_2+R_1M_1+R_2,
\end{aligned}
\end{equation}
we obtain
\begin{equation}
\begin{aligned}
\gamma^2_{N,N}(u)&=[M_1]_{12}[M_1]_{21}+[M_1]_{12}[R_1]_{21}+[M_1]_{21}[R_1]_{12}+[R_1]_{12}[R_1]_{21},\\
&=[M_1]_{12}[M_1]_{21}+\mathcal{O}(N^{-1/5}),\\
&=\frac{(1-\sigma_u)^2}{16}+\mathcal{O}(N^{-1/5}),\\
\beta_{N,N}(u)&=\frac{[M_2]_{12}+[R_1M_1]_{12}+[R_2]_{12}}
{[M_1]_{12}+[R_1]_{12}}-[M_1]_{22}+[R_1]_{22}\\
&=\frac{[M_2]_{12}}{[M_1]_{12}}-[M_1]_{22}+\mathcal{O}(N^{-1/5})\\
&=\frac{1+\sigma_u}{2}+\mathcal{O}(N^{-1/5}).
\end{aligned}
\end{equation}

More generally, because of the expansion of $R(\z)$ in inverse powers of $N^{1/5}$ obtained in Proposition \ref{propRn}, the coefficients $\gamma^2_{N,N}(u)$ and $\beta_{N,N}(u)$ can be expanded in inverse powers of $N^{1/5}$ as well. As a consequence of the Riemann--Hilbert analysis, we have the following asymptotic expansions:
\begin{equation}\label{asympgnbn}
\begin{aligned}
\gamma^2_{N,N}(u)&\sim \ga_c^2+\sum_{k=1}^{\infty}\frac{1}{N^{k/5}}p_k(\lambda),\qquad
\beta_{N,N}(u)\sim \be_c+\sum_{k=1}^{\infty}\frac{1}{N^{k/5}}q_k(\lambda),
\end{aligned}
\end{equation}
where $\ga_c^2$ and $\be_c$ are the values of the coefficients at the critical time $u=u_c$, and $u$ is related to $\lambda$ by means of the double scaling relation $N^{4/5}(u-u_c)=c_1\lambda$. 

In order to prove that the odd terms in the previous expansions are $0$, as claimed in the theorem, and also the connection with the solution to the Painlev\'e I differential equation, it is more convenient to use the double scaling in terms of the ratio $n/N$, see \eqref{nNv}. 

We can establish a simple analytic relation between the variables $\lambda$ and $v$ in a neighborhood of the origin, which corresponds to the critical case.
\begin{lemma} We can write
\begin{equation}
\lambda(v)=N^{4/5}\varphi(vN^{-4/5}),
\end{equation}
where $\varphi(t)$ is analytic in a neighborhood of $t=0$ and $\varphi(0)=0$. Furthermore, $\varphi'(0)\neq 0$, so
\begin{equation}
v(\lambda)=N^{4/5}\varphi^{-1}(\lambda N^{-4/5}),
\end{equation}
the inverse function $\varphi^{-1}(t)$ being analytic in a neighborhood of $t=0$.
\end{lemma}
\begin{proof}
We can express the variable $\lambda$ in terms of $v$:
\begin{equation}
1+vN^{-4/5}=\frac{n}{N}=\frac{u^2}{u_c^2}=(1+c_2\lambda N^{-4/5})^2
\end{equation}
where $c_2=c_1/u_c=2^{-7/5}$, using \eqref{dsc_u}, and hence for bounded $v$ and large enough $N$ we can take square root and write:

\begin{equation}\label{xv}
\lambda(v)=N^{4/5} \varphi(N^{-4/5}v), 
\end{equation}
where the function
\begin{equation}\label{phit}
\varphi(t)=\frac{\sqrt{1+t}-1}{c_2}
\end{equation}
is analytic in a neighborhood of $t=0$. The first few terms in the Taylor expansion give
\begin{equation}\label{xvexplicit}
\lambda(v)=\frac{v}{2c_2}-\frac{v^2}{8c_2}N^{-4/5}+
\mathcal{O}(N^{-8/5}),
\end{equation} 
so $\varphi'(0)=2^{2/5}\neq 0$, and by the inverse function theorem, locally near $\lambda=0$ we can write
\begin{equation}\label{vx}
v(\lambda)=N^{4/5}\varphi^{-1}(N^{-4/5}\lambda),
\end{equation}
where $\varphi^{-1}(t)$ is analytic in a neighborhood of $t=0$. By perturbation of \eqref{xvexplicit}, the first few terms of this expansion are
\begin{equation}\label{vxexplicit}
v(\lambda)=2c_2\lambda+c_2^2 \lambda^2 N^{-4/5}+\mathcal{O}(N^{-8/5}).
\end{equation}

\end{proof}

%

Now consider the recurrence coefficients corresponding to the new
weight $e^{-NV(\xi;u_c)}$, where
$V(\xi;u_c)=\frac{\xi^2}{2}-u_c\xi^3$. We denote these coefficients
by $\ga^2_{N,n}$ and $\be_{N,n}$, and we have
\begin{equation}\label{newcoeffs}
\begin{aligned}
\ga_{N,n}^2(\lambda)&=\frac{u^2}{u_c^2}\ga_{N,N}^2(\lambda)\sim \left[\ga^2_c+\sum_{k=1}^{\infty}\frac{1}{N^{k/5}}\,p_{k}(\lambda)\right]\left(1+\frac{c_2\lambda}{N^{4/5}}\right)^2,\\
\be_{N,n}(\lambda)&=\frac{u}{u_c}\be_{N,N}(\lambda)\sim \left[\be_c+\sum_{k=1}^{\infty}\frac{1}{N^{k/5}}\,q_{k}(\lambda)\right]\left(1+\frac{c_2\lambda}{N^{4/5}}\right).
\end{aligned}
\end{equation}

These asymptotic expansions are a consequence of the Riemann--Hilbert analysis, since we can expand the coefficients $\ga^2_{N,N}(u)$ and $\be_{N,N}(u)$ in inverse powers of 
$N^{1/5}$, see \eqref{asympgnbn}.


Using the double scaling and the relation between $\lambda$ and $v$ from the previous lemma, we can write an expansion in inverse powers of $N^{1/5}$:
\begin{equation}\label{asympgNbNv}
\begin{aligned}
\ga_{N,n}^2(v)&\sim \left[\ga^2_c+\sum_{k=1}^{\infty}\frac{p_k(N^{4/5} \varphi(vN^{-4/5}))}{N^{k/5}(1+vN^{-4/5})^k}\,\right]
\left(1+\frac{c_2\varphi(vN^{-4/5})}{(1+vN^{-4/5})^{4/5}}\right)^2=\ga^2_c+\sum_{k=1}^{\infty}\frac{1}{N^{k/5}}\,\hat{p}_{k}(v),\\
\beta_{N,n}(v) & \sim \left[\be_c+\sum_{k=1}^{\infty}\frac{q_k(N^{4/5} \varphi(vN^{-4/5}))}{N^{k/5}(1+vN^{-4/5})^k}\right]\left(1+\frac{c_2\varphi(vN^{-4/5})}{(1+vN^{-4/5})^{4/5}}\right)=\be_c+\sum_{k=1}^{\infty}\frac{1}{N^{k/5}}\hat{q}_{k}(v),
\end{aligned}
\end{equation}
for some coefficients $\hat{p}_{k}(v)$ and $\hat{q}_{k}(v)$ that can be computed from the previous formula, in terms of $p_{k}$ and $q_k$, using the analyticity of the coefficients $p_k$ and $q_k$ and of the function $\varphi$. Explicitly, the first coefficients are
\begin{equation}
\begin{aligned}
\hat{p}_k(v)&=p_k\left(\frac{v}{2c_2}\right)=p_k(2^{-2/5}v), \qquad k=1,2,3,4,\\
\hat{q}_k(v)&=q_k\left(\frac{v}{2c_2}\right)=q_k(2^{-2/5}v), \qquad k=1,2,3,4.
\end{aligned}
\end{equation}

Higher order coefficients can be computed, using $p_k$, $q_k$ and
their derivatives. Next we consider the asymptotic expansions
\eqref{asympgNbNv} together with the string equations corresponding
to the weight $e^{-NV(\z;u_c)}$ (we omit the dependence on $u_c$ for
brevity):
%
%
%

\begin{equation}\label{string2}
\begin{aligned}
3u_c(\ga^2_{N,n}+\be^2_{N,n}+\ga^2_{N,n+1})&=\be_{N,n},\\
\ga^2_{N,n}(1-3u_c(\be_{N,n}+\be_{N,n-1}))&=\frac{n}{N}.
\end{aligned}
\end{equation}

We note the symmetry of the first string equation with respect to the change of indices
\begin{equation}
\sigma_0=\{\ga^2_{N,j}\rightarrow \ga^2_{N,2n+1-j},
\be_{N,j}\rightarrow \be_{N,2n-j},j=0,1,\ldots\},
\end{equation}
and the second string equation with respect to
\begin{equation}
\sigma_1=\{\ga^2_{N,j}\rightarrow \ga^2_{N,2n-j},\be_{N,j}\rightarrow \be_{N,2n-1-j},\,\, j=0,1,2,\ldots
\}.
\end{equation}

As a consequence, all odd terms in the expansion are $0$, and in
fact we can write
\begin{equation}\label{asympgNbNv2N}
\begin{aligned}
\gamma_{N,n}^2(v)&\sim \ga^2_c+\sum_{k=1}^{\infty}\frac{1}{N^{2k/5}}\,\hat{p}_{2k}(v),\qquad
\beta_{N,n}(v)  \sim \be_c+\sum_{k=1}^{\infty}\frac{1}{N^{2k/5}}\,\hat{q}_{2k}(\tilde{v}).
\end{aligned}
\end{equation}

The complete proof of this property can be found in \cite[Section 5]{BI}.

In order to identify the coefficients $\hat{p}_2$ and $\hat{q}_2$ in terms of Painlev\'e I functions, we take the first few terms from the asymptotic expansion \eqref{asympgNbNv2N}:
\begin{equation}\label{cr12}
\begin{aligned}
\ga_{N,n}^2(v)&=\ga_c^2+N^{-2/5}\hat{p}_2(v)+N^{-4/5}\hat{p}_{4}(v)+\ldots,\\
\be_{N,n}(v)&=\be_c+N^{-2/5}\hat{q}_2(\tilde{v})+N^{-4/5}\hat{q}_{4}(\tilde{v})+\ldots,
\end{aligned}
\end{equation}

Now we substitute this into the string equations \eqref{string2} and consider different powers of $N$. At the zeroth level we have the following relations for $\hat{p}_0=\ga_c^2$ and $\hat{q}_0=\be_c$:
\begin{equation}\label{g0b0}
\begin{aligned}
72u_c^2\hat{p}^3_0-\hat{p}^2_0+1&=0, \qquad
\hat{q}_0=\frac{\hat{p}_0-1}{6u_c\hat{p}_0}.
\end{aligned}
\end{equation}

The cubic equation for $\hat{p}_0$ has a double root
\begin{equation}\label{gc}
\ga^2_c=\sqrt{3}
\end{equation}
and one simple root $\ga^2_{c*}=-\tfrac{\sqrt{3}}{2}$. The critical value for $\hat{q}_0$ can be computed from the second equation in \eqref{g0b0}:
\begin{equation}\label{bc}
\be_c=\frac{\ga^2_c-1}{6u_c\ga^2_c}=3^{1/4}(\sqrt{3}-1).
\end{equation}

These critical values can also be computed from the value of the coefficients $g_0(u_c)$ and $b_0(u_c)$ coming from the regular regime, see \cite[\S 7]{BD}.

Next we have factors multiplying $N^{-2/5}$. We consider the first string equation and expand around $\tilde{v}$ Observe that
\begin{equation}\label{cr15}
\begin{aligned}
\ga_{N,n}^2(\tilde{v})&=\ga^2_c+N^{-2/5} \hat{p}_2\left(\tilde{v}-\frac{N^{-1/5}}{2}\right)+N^{-4/5} \hat{p}_4\left(\tilde{v}-\frac{N^{-1/5}}{2}\right)+\ldots,\\
\ga_{N,n+1}^2(\tilde{v})&=\ga^2_c+N^{-2/5} \hat{p}_2\left(\tilde{v}+\frac{N^{-1/5}}{2}\right)+N^{-4/5} \hat{p}_4\left(\tilde{v}+\frac{N^{-1/5}}{2}\right)+\ldots,
\end{aligned}
\end{equation}
so when we expand in Taylor series around $\tilde{v}$ and add these two quantities, only the even terms remain:
\begin{equation}
\begin{aligned}
\ga_{N,n}^2+\ga^2_{N,n+1}&=
2\ga^2_c+2N^{-2/5}\hat{p}_2(\tilde{v})+ N^{-4/5}\left(\frac{1}{4}\hat{p}''_2(\tilde{v})
+2\hat{p}_4(\tilde{v})\right)+\ldots,
\end{aligned}
\end{equation}
and then we have from the first string equation
\begin{equation}
6u_c\hat{p}_2(\tilde{v})-(1-6u_c\beta_c)\hat{q}_2(\tilde{v})=0.
\end{equation}

Since $1-6u_c\be_c=1/\sqrt{3}$, see \eqref{bc}, we obtain
\begin{equation}
\hat{q}_2(\tilde{v})=3^{-1/4}\,\hat{p}_2(\tilde{v}).
\end{equation}

Now we want to balance the terms multiplying $N^{-4/5}$, and   that gives
\begin{equation}\label{ode1}
\frac{1}{8}\hat{p}''_2(\tilde{v})+\frac{1}{2}\hat{q}^2_2(\tilde{v})
=-\hat{p}_4(\tilde{v})+3^{1/4}\hat{q}_4(\tilde{v}).
\end{equation}

Consider now the second string equation in \eqref{string2}, we expand around $v$. From \eqref{cr12} we have 
\begin{equation}\label{cr27}
\begin{aligned}
\be_{N,n}(v)&=\be_c+N^{-2/5}\hat{q}_2\left(v+\frac{N^{-1/5}}{2}\right)+N^{-4/5}\hat{q}_4\left(v+\frac{N^{-1/5}}{2}\right)+\ldots\\
\be_{N,n-1}(v)&=\be_c+N^{-2/5}\hat{q}_2\left(v-\frac{N^{-1/5}}{2}\right)+N^{-4/5}\hat{q}_4\left(v-\frac{N^{-1/5}}{2}\right)+\ldots,
\end{aligned}
\end{equation}
hence
\begin{equation}\label{cr29}
\begin{aligned}
\be_{N,n}(v)+\be_{N,n-1}(v)
&=2\be_c+2N^{-2/5}\hat{q}_2(v)
+N^{-4/5}\left(\frac{1}{4}\hat{q}_2''(v)+2\hat{q}_4(v)\right)+\ldots
\end{aligned}
\end{equation}

By substituting this expression into the second equation in \eqref{string2}, we obtain the following equation for the terms multiplying $n^{-2/5}$:
\begin{equation}
(1-6u_c\be_c)\hat{p}_2(v)-6u_c\ga^2_c\hat{q}_2(v)=0,
\end{equation}
which gives $\hat{q}_2(v)=3^{-1/4}\,\hat{p}_2(v)$ again. If we group terms that multiply $N^{-4/5}$, we get
\begin{equation}\label{ode2}
-\frac{3^{1/4}}{8}\hat{q}_2''(v)-3^{-1/4}\,\hat{p}_2(v)\hat{q}_2(v)-3^{1/2}v=-\hat{p}_4(v)+3^{1/4}\hat{q}_4(v).
\end{equation}

Let us compare equations \eqref{ode1} and \eqref{ode2}. These are two equations on the
functions $\hat{p}_2(v)$, $\hat{q}_2(v)$, $\hat{p}_4(v)$, and $\hat{q}_4(v)$. The fact that in \eqref{ode1} we use
the variable $\tilde v$, and not $v$, does not matter, because these are equations on functions, and we can replace $\tilde v$ by $v$ in \eqref{ode1}. Since the right hand sides
of equations \eqref{ode1}, \eqref{ode2} are equal, the left hand sides are equal as well:
\begin{equation}\label{cr39}
-\frac{3^{1/4}}{8}\hat{q}_2''(v)-3^{-1/4}\,\hat{p}_2(v)\hat{q}_2(v)-3^{1/2}v=\frac{1}{8}\hat{p}''_2(\tilde{v})+\frac{1}{2}\hat{q}^2_2(\tilde{v}).
\end{equation}

We can reduce this to an equation on $\hat{p}_2(v)$ only:
\begin{equation}\label{cr40}
\hat{p}_2''(v)=-2\sqrt{3}\, \hat{p}_2^2(v)-4\sqrt{3}\, v,
\end{equation}
which is the Painlev\'e I equation. Applying the change of variables
\begin{equation}\label{cr41}
v=2^{-2/5}\lambda,\qquad \hat{p}_2(v)=-2^{4/5}3^{1/2} y_{\alpha}(\lambda),
\end{equation}
we bring it to the standard form:
\begin{equation}\label{cr42}
y_{\alpha}''(\lambda)=6 y_{\alpha}^2(\lambda)+\lambda.
\end{equation}

We know that $y_{\alpha}(\lambda)$ is the appropriate solution of Painlev\'e I because of the asymptotic behavior coming from the Riemann--Hilbert analysis. 

The asymptotic expansions for the coefficients $\ga_{N,N}^2(\lambda)$ and $\be_{N,N}(\lambda)$ in inverse powers of
$N^{2/5}$ now follow directly from \eqref{asympgNbNv2N} and \eqref{newcoeffs}:

\begin{equation}
\begin{aligned}
\ga^2_{N,N}(\lambda)&\sim \left[\ga^2_c+\sum_{k=1}^{\infty}\frac{\hat{p}_{2k}(v)}{N^{2k/5}}\right]\frac{1}{(1+c_2\lambda N^{-4/5})^2}
=\ga^2_c+\sum_{k=1}^{\infty}\frac{p_{2k}(\lambda)}{N^{2k/5}},
\end{aligned}
\end{equation}
for some coefficients $p_{2k}(\lambda)$. Here we have used the change of variables from $v$ to $\lambda$ and also
\begin{equation}
\frac{n}{N}=1+vN^{-4/5}=1+\varphi^{-1}(\lambda N^{-4/5}).
\end{equation}

The first coefficients are
\begin{equation}\label{pphat}
p_{2}(\lambda)=\hat{p}_{2}(2c_2 \lambda)=\hat{p}_2(2^{-2/5}\lambda), \qquad
p_{4}(\lambda)=\hat{p}_{4}(2c_2 \lambda)=\hat{p}_4(2^{-2/5}\lambda).
\end{equation}

Similarly,
\begin{equation}
\begin{aligned}
\be_{N,N}(\lambda)&\sim \left[\be_c+\sum_{k=1}^{\infty}\frac{\hat{q}_{2k}(v)}{N^{2k/5}}\right]\frac{1}{1+c_2\lambda N^{-4/5}}
\sim \be_c+\sum_{k=1}^{\infty}\frac{q_{2k}(\lambda)}{N^{2k/5}},
\end{aligned}
\end{equation}
for some coefficients $q_{2k}(\lambda)$. The first terms are
\begin{equation}\label{qqhat}
q_{2}(\lambda)=\hat{q}_{2}(2c_2 \lambda)=\hat{q}_2(2^{-2/5}\lambda), \qquad
q_{4}(\lambda)=\hat{q}_{4}(2c_2 \lambda)=\hat{q}_4(2^{-2/5}\lambda).
\end{equation}

Finally, from \eqref{pphat}, we have 
\begin{equation}
 \hat{p}_2(2^{-2/5}\lambda)=p_2(\lambda)=-2^{4/5}3^{1/2} y(\lambda), 
\end{equation}
and also $q_2(\lambda)=\hat{q}_2(2^{-2/5}\lambda)=3^{-1/4} \hat{p}_2(2^{-2/5}\lambda)=3^{-1/4}p_2(\lambda)$, which concludes the proof.

\section{Proof of Theorem \ref{th_extended}}\label{sec_free}

From \cite{BD}, we know that in the regular case (when $u_c-u\geq
\delta$, for some fixed $\delta$ independent of $N$), the
coefficient  $\ga_{N,N}^2(u)$ admits an asymptotic expansion in
inverse powers of $N^{2}$. On the other hand, Theorem \ref{Th1}
gives an asymptotic expansion in inverse powers of $N^{2/5}$ in the
double scaling regime. The issue now is how to cover the gap between the regular and the
double scaling cases in the asymptotics of $\ga^2_{N,N}(u)$, by extending the regions of validity of those asymptotic expansions.

\subsection{Extension of the regular regime}

In the regular case, in terms of the parameter $w=u^2$ and setting $n=N$, we have
the following asymptotic expansions, see \cite{BD}:
\begin{equation}\label{asympgbhat}
\begin{aligned}
\hat{\gamma}_{N,N}^2(w)& \sim \sum_{k=0}^{\infty}\frac{u^{4k}}{N^{2k}}\,\hat{g}_{2k}(w), \qquad 
\hat{\beta}_{N,N}(w) \sim \sum_{k=0}^{\infty}\frac{u^{4k}}{N^{2k}}\,\hat{b}_{2k}(w),
\end{aligned}
\end{equation}
where
\begin{equation}\label{hatbg_1}
\begin{aligned}
\hat{g}_{2k}(w)=u^{-4k+2}g_{2k}(u),\qquad
\hat{b}_{2k}(w)=u^{-4k+1}b_{2k}(u), \qquad k\geq 0.
\end{aligned}
\end{equation}

These asymptotic expansions are uniform in $u$ if $0< u\leq
u_c-\delta$, for fixed $\delta>0$. This is a consequence of the
Riemann--Hilbert analysis. The purpose of this section is to extend
this property to the case when we let $u-u_c$ decrease with $N$, at
a certain (slow enough) rate. To this end, we will need a local
analysis near $\zeta=1$ in a neighborhood that shrinks with $N$.

In the regular case, we recall the global parametrix:
\begin{equation}
M(\z)=\frac{1}{2}
\begin{pmatrix}
1& 1\\ i& -i
\end{pmatrix}
\begin{pmatrix}
\beta^{-1}(\z)& 0\\ 0 & \beta(\z)
\end{pmatrix}
\begin{pmatrix}
1& 1\\ i& -i
\end{pmatrix},
\qquad \beta(\z)=\left(\frac{\z+1}{\z-1}\right)^{1/4}.
\end{equation}

We construct the local parametrix $P(\z)$ in a neighborhood of $\z=1$ in terms of Airy functions. Namely, given a disc $D(1,\varepsilon)$ and $\Gamma_S$ the union of contours shown in Figure \ref{GaS_extended}, the function $P(\z)$ satisfies the following RH problem:
\begin{enumerate}
\item $P(\z)$ is analytic in $D(1,\varepsilon)\setminus \Gamma_S$, and for every $x\in\Gamma_S$ the following limits exist:
\begin{equation}\label{ps1}
P_{\pm}(s)=\lim_{\z\to s,\; \z\in \Om_{\pm}}P(\z),
\end{equation}
\item On $\Gamma_S$, we have the following jumps:
\begin{equation}\label{ps2}
P_{+}(s)=P_{-}(s)
\begin{cases}
 \begin{pmatrix}
  1 & e^{N\phi_u(s)}\\ 0 & 1
 \end{pmatrix},
& s\in D(1,\varepsilon)\cap [1,\infty)\\
 \begin{pmatrix}
  1 & 0\\ e^{-N\phi_u(s)} & 1
 \end{pmatrix},
& s\in D(1,\varepsilon)\cap (\gamma_+\cup\gamma_-)\\
\begin{pmatrix}
  0 & 1\\-1 & 0
 \end{pmatrix},
& s\in D(1,\varepsilon)\cap [\sigma_u,1]\\
\end{cases}
\end{equation}
\item Uniformly for $\z\in\partial D(1,\varepsilon)$,
\begin{equation}\label{ps3}
P(\z)=M(\z)\left(I+\mathcal O\left(\frac{1}{N}\right)\right).
\end{equation}
\end{enumerate}

\begin{figure}
\centerline{\includegraphics{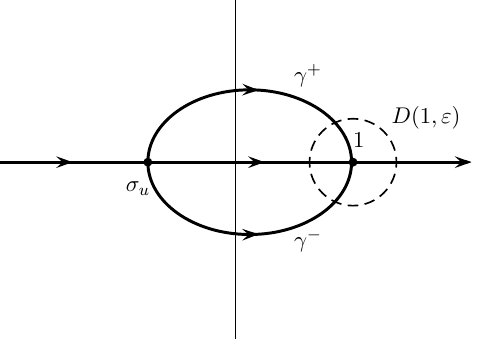}}
\caption{Contour $\Gamma_S$ for the extension of the regular regime.}
\label{GaS_extended}
\end{figure}

Define 
\begin{equation}
 \tilde{P}(\z)=P(\z)e^{N\phi_u(\z)\sigma_3/2},
\end{equation}
then we have
\begin{enumerate}
\item $\tilde{P}(\z)$ is analytic in $D(1,\varepsilon)\setminus \Gamma_S$, and for every $x\in\Gamma_S$ the following limits exist:
\begin{equation}\label{pu1}
\tilde{P}_{\pm}(\z)=\lim_{s\to\z,\; \z\in \Om_{\pm}}\tilde{P}(s),
\end{equation}
\item On $\Gamma_S$, we have the following jumps:
\begin{equation}\label{pu2}
\tilde{P}_{+}(\z)=\tilde{P}_{-}(\z)
\begin{cases}
 \begin{pmatrix}
  1 & 1\\ 0 & 1
 \end{pmatrix},
& \z\in D(1,\varepsilon)\cap [1,\infty)\\
 \begin{pmatrix}
  1 & 0\\ 1 & 1
 \end{pmatrix},
& \z\in D(1,\varepsilon)\cap (\gamma_+\cup\gamma_-)\\
\begin{pmatrix}
  0 & 1\\-1 & 0
 \end{pmatrix},
& \z\in D(1,\varepsilon)\cap [\sigma_u,1]\\
\end{cases}
\end{equation}
\item Uniformly for $\z\in\partial D(1,\varepsilon)$,
\begin{equation}\label{pu3}
\tilde{P}(\z)=M(\z)\left(I+\mathcal O\left(\frac{1}{N^{\varepsilon_3}}\right)\right) e^{N\phi_u(\z)\sigma_3/2},
\end{equation}
for $\varepsilon_3>0$.
\end{enumerate}

Consequently, we look for a local parametrix in the form
\begin{equation}
\tilde{P}(\z)=E(\z)A(N^{2/3}f(\z)),
\end{equation}
with the analytic prefactor
\begin{equation}
E(\z)=M(\z)\frac{1}{\sqrt{2}}\begin{pmatrix} 1 & i\\ i &
1\end{pmatrix}^{-1}(N^{2/3}f(\z))^{\sigma_3/4}.
\end{equation}

The matrix $A(N^{2/3}f(\z))$ solves the standard
Airy--Riemann--Hilbert problem in an auxiliary $w=N^{2/3}f(\z)$
plane, and the function $f(\z)$ is
\begin{equation}
f(\z)=\left(-\frac{3}{4}\phi_u(\z)\right)^{2/3}=\left(-\frac{3\pi i}{2}\int_{\z}^1 \varrho_u(s)ds\right)^{2/3}, \qquad \z\in\mathbb{C}\setminus(-\infty,1],
\end{equation}
where $\varrho_u(s)$ is the density of the equilibrium measure given by \eqref{rhoz}. 

If we expand in powers of $s-1$ and integrate, we obtain
\begin{equation}\label{phiuat1}
\begin{aligned}
\phi_u(\z)
&=C_u(\z-1)^{3/2}\left(1+\frac{3\z_0-15}{20(\z_0-1)}(\z-1)+\mathcal{O}((\z-1)^{2})\right),
\end{aligned}
\end{equation}
where
\begin{equation}
 C_u=\frac{3\sqrt{2}}{16}u(b-a)^3(\z_0-1)>0.
\end{equation}

Recall that $\z_0>1$ if $0\leq u<u_c$, as a consequence of \cite[Prop. 2.2]{BD}. If $\z_0=1+\delta$, with $\delta>0$, then the root of the linear factor in the previous expansion is
\begin{equation}
 \z^*=1-\frac{20(\z_0-1)}{3\z_0-15}=1+\frac{5}{3}\delta+\mathcal{O}(\delta^2)
\end{equation}
  
This proves that $f(\z)$ is a conformal mapping of a neighborhood of $\z=1$ onto a neighborhood of the origin in the $w$ plane.  

If we define now $R(\z)=S(\z)M(\z)^{-1}$ outside the discs
$D(\pm1,\varepsilon)$, and $R(\z)=S(\z)\tilde{P}(\z)^{-1}$ inside,
then the jump for this matrix $R(\z)$ on the boundary of the discs
is $J_R(\z)=M(\z) \tilde{P}(\z)^{-1}$, and because of the matching
between $\tilde{P}(\z)$ and $M(\z)$, this jump is equal to
$I+\mathcal{O}(1/N)$. This is used to prove that $R(\z)$ itself can
be expanded in inverse powers of $N$, beginning with $I$, and this
in turn is used to prove that the recurrence coefficients
$\ga^2_{N,N}(u)$ and $\be_{N,N}(u)$ admit an asymptotic expansion in powers
of $1/N$ too. This expansion is uniform in $u$ for $0<u\leq
u_c-\delta$.

It is clear that if $\z_0=1+\delta$, for some fixed $\delta>0$, then we can find a (fixed) neighborhood of $\z=1$ where the linear factor in \eqref{phiuat1} is positive. The problem that we will encounter when trying to extend the regular case closer to the critical value $u=u_c$ is that the single root of the density $\varrho_u(\z)$ can get arbitrarily close to $\z=1$, and that spoils the analyticity of $f(\z)$ in a neighborhood of $\z=1$. More precisely, if $\z_0=1+cN^{-\ga}$, for some $\ga>0$ and $c>0$, then the linear factor before will vanish at the point
\begin{equation}
\z^*=1+\frac{5c}{4}N^{-\ga}+\mathcal{O}(N^{-2\ga}),
\end{equation}
and $f(\z)$ will not be analytic in any fixed neighborhood of $\z=1$ for $N$ large enough. To get analyticity in this context, we construct a shrinking neighborhood of $\z=1$.

We make the change of variables $s=1+\xi N^{-\ga}$ and $\z=1+\tau N^{-\ga}$, then we obtain
\begin{equation}
\phi_u(\tau)=
C_u N^{-5\ga/2}\tau^{3/2}\left(\tau-\frac{20c}{12-3cN^{-\ga}}+\mathcal{O}(\tau^{2})\right).
\end{equation}

Now the linear factor is harmless when $c>0$, since for any $N\geq 1$ and $\ga>0$ we can bound
\begin{equation}
\frac{20c}{12-3cN^{-\ga}}>\frac{5c}{3},
\end{equation}
and then in any disc around $\tau=0$ of radius $\delta<\tfrac{5c}{3}$, the function
\begin{equation}
N^{2/3}f(\tau)=N^{2/3}\left(-\frac{3}{4}\phi_u(\tau)\right)^{2/3}
=C_u^{3/2} N^{\frac{2-5\ga}{3}}
\tau\left(\tau-\frac{20c}{12-3cN^{-\ga}}+\mathcal{O}(\tau^{2})\right)^{2/3}
\end{equation}
will be analytic.

Because of the shrinking neighborhood, we get an extra factor $N^{\frac{2-5\ga}{3}}$. In order to be able to do the matching with the global parametrix, we need the argument of the Airy function to grow large with $N$, and that leads to the condition
\begin{equation}\label{restga1}
\frac{2-5\ga}{3}>0 \Rightarrow
\ga<\frac{2}{5}.
\end{equation}

Hence, we suppose that $\gamma=\frac{2}{5}-\frac{\varepsilon_2}{2}$, with $\ep_2>0$, so $\z_0=1+c N^{-\frac{2}{5}+\frac{\ep_2}{2}}$. Bearing in mind \eqref{z0s} and the fact that the variables $u$ and $s$ are essentially equivalent near the critical value, we conclude that we can cover the following range with the regular case:
\begin{equation}\label{extreg}
u-u_c\sim \tilde{c}_3 N^{-\frac{4}{5}+\ep_2}.
\end{equation}

We note that in this setting we still obtain an asymptotic expansion for $R(\z)$ that is uniform in $u$, provided that \eqref{extreg} is satisfied. The difference is that the coefficients in the asymptotic expansion coming from the Airy parametrix will now depend on $N$, and actually grow as $N$ gets large:
\begin{equation}
A(N^{2/3}f(\z))=(N^{2/3}f(\z))^{-\sigma_3/4}\frac{1}{\sqrt{2}}
\begin{pmatrix}
1 & i \\ i & 1
\end{pmatrix}
\left(I+\sum_{k=1}^{\infty}\frac{a_k N^{k(5\ga/2-1)}}{\tau^{3k/2}}\right)e^{-N\phi_u(\z)\sigma_3/2}.
\end{equation}

However, because of the restriction \eqref{restga1}, the exponent of
$N$ is still $<0$ for every $k\geq 1$, so we obtain an asymptotic
series, uniform in $u$ in this extended regular regime. Note that the term in parenthesis is
\begin{equation}
I+\mathcal{O}(N^{5\gamma/2-1})=I+\mathcal{O}(N^{-5\varepsilon_2/4}) 
\end{equation}
if $\gamma=\frac{2}{5}-\frac{\varepsilon_2}{2}$, with $\varepsilon_2>0$, and this has an effect on the matching condition \eqref{pu3}.

\subsection{Extension of the double scaling regime}

In the double scaling regime, the asymptotic expansion for
$\ga^2_{N,N}(u)$ was obtained applying the Deift--Zhou nonlinear
steepest descent to the Riemann--Hilbert problem. In particular, in
the final step, the form of the asymptotic expansion for
$R(\zeta)$ comes from the local parametrix near $\z=1$, given in
terms of the solution to the Painlev\'e I equation.

Observe that the uniform asymptotic expansion for $\Phi(\z;\lambda,\alpha)$ in \eqref{asymp_Phi_final} is obtained by considering the function $g(\z)$ instead of $\theta_0(\z)$. Bearing in mind \eqref{rhp19A}, the first term that we neglect is 
\begin{equation}
a_0(\z,\lambda)=-\frac{\sqrt{6}}{9}(-\lambda)^{5/4}\z^{-1/2},
\end{equation}
and we need this to be small as we let $(-\lambda)\to\infty$. In the construction of the local parametrix, see \eqref{px9}, we have $\z=\mathcal{O}(N^{2/5})$, and we suppose now that $-\lambda=\mathcal{O}(N^{\varepsilon})$, for some $\varepsilon>0$. In order for $a_0(\z,\lambda)$ to tend to $0$ as $N\to\infty$, we need
\begin{equation}
\frac{5\varepsilon}{4}-\frac{1}{5}<0 \Rightarrow \varepsilon<\frac{4}{25}, 
\end{equation}
so the extension of the double scaling regime is possible provided that $(-\lambda)=\mathcal{O}(N^{4/25-\varepsilon})$ as $N\to\infty$. As a consequence, one can allow $u-u_c=\mathcal{O}(N^{-16/25-\varepsilon})$, as $N\to\infty$, so there is indeed an overlap between both regimes.

\section{Proof of Theorem \ref{thfree}}

In order to analyze the free energy near the critical case, we will
use the Toda equation, similarly to \cite{BD}:
\begin{equation}\label{Toda}
\frac{d^2 \tilde{F}_N(t)}{dt^2}=\tilde{\ga}^2_{N,N}(t),
\end{equation}
expressed in terms of the parameter $t$, which is related to $u$ as
follows:
\begin{equation}\label{tu}
t=\frac{1}{4(3u)^{4/3}}.
\end{equation}

Accordingly, the new critical value is $t_c=3\cdot 2^{-2/3}$. We recall from \cite[Proposition 5.1]{BD} that the Toda equation \eqref{Toda} holds for any $t>t_c$.

The corresponding weight function is
$\tilde{V}(w;t)=-\frac{w^3}{3}+tw$. More precisely, if we
apply the change of variables
\begin{equation}\label{zw}
z=(3u)^{-1/3}w+\frac{1}{6u}\,,
\end{equation}
where we assume that $u>0$ and $(3u)^{-1/3}>0$, we have, by
straightforward algebra, that
\begin{equation}\label{wz1}
V(z;u)-\frac{1}{108u^2}=\tilde{V}(w;t).
\end{equation}

We observe that the recurrence coefficients in the $t$ parameter,
that we denote $\tilde{\ga}_{N,N}(t)$ and $\tilde{\be}_{N,N}(t)$,
can be obtained from the original ones in the $u$ parameter as
follows:
\begin{equation}\label{gbu}
\begin{aligned}
\tilde{\gamma}_{N,N}^2(t)&=(3u)^{2/3}\gamma_{N,N}^2(u)=\frac{\gamma_{N,N}^2(u)}{2\sqrt{t}}\,,\\
\tilde{\beta}_{N,N}(t)&=(3u)^{1/3}\left(\beta_{N,N}(u)-\frac{1}{6u}\right)=\frac{1}{(4t)^{1/4}}\left(\beta_{N,N}(u)-\frac{1}{6u}\right).
\end{aligned}
\end{equation}

We have the following relation between the free energy in the variables $u$ and $t$:
\begin{equation}\label{FNFN}
\tilde{F}_N(t)=\frac{1}{108 u^2}+\frac{\ln
\,(3u)}{3}+F_N(u)=\frac{2\,t^{3/2}}{3}-\frac{\ln\,(4t)}{4}+F_N(u)\,.
\end{equation}

We start from the regular regime, which can be extended to cover the range $u_c-u=\mathcal{O}(N^{-4/5+\varepsilon})$. In this case
\begin{equation}
 \tilde{F}_N(t)=\tilde{F}_0(t)+\frac{\tilde{F}_2(t)}{N^2}+\mathcal{O}(N^{-4}).
\end{equation}

The behavior of $\tilde{F}_N(t)$ near the critical point can be analyzed as in \cite{BD}, using the Toda equation \eqref{Toda}, bearing in mind that $\tilde{\gamma}^2_{N,N}(t)$ admits an asymptotic expansion in powers of $N^{-2}$ if $t>t_c$:
\begin{equation}
\tilde{\gamma}^2_{N,N}(t)=\tilde{g}_0(t)+\frac{\tilde{g}_2(t)}{N^2}+\mathcal{O}(N^{-4}).
\end{equation}

The first term $\tilde{g}_0(t)$ behaves as follows near $t=t_c$:
\begin{equation}
\tilde{g}_0(t)=2^{-2/3}-2^{-1/3}3^{-1/2} (\Delta t)^{1/2}+\mathcal{O}(\Delta t), \qquad \Delta t=t-t_c.
\end{equation}

Integrating twice \eqref{Toda} from $t=t_c$, we obtain 
\begin{equation}
\tilde{F}^{(0)}(t)=-\frac{2^{5/3}3^{1/2}}{45}(\Delta t)^{5/2}+2^{-5/3}(\Delta t)^2+\tilde{A}+\tilde{B}\cdot \Delta t+\mathcal{O}((\Delta t)^3),
\end{equation}
where $\tilde{A}$ and $\tilde{B}$ are constants. The next term is
\begin{equation}
\tilde{F}^{(2)}(t)=-\frac{1}{48}\ln (\Delta t)+\tilde{D}+\mathcal{O}((\Delta t)^{1/2}),
\end{equation}
where $\tilde{D}$ is a constant. We add and subtract the logarithmic term, which is singular at $t=t_c$:
\begin{equation}
\tilde{F}_N(t)=\tilde{F}^{(0)}(t)+\frac{\tilde{F}^{(2)}(t)+\frac{1}{48}\ln(\Delta t)}{N^2}-\frac{\ln(\Delta t)}{48N^2}+\mathcal{O}(N^{-4})
\end{equation}

We now define
\begin{equation}
 \tilde{F}^{\operatorname{reg}}_N(t)=2^{-5/3}(\Delta t)^2+\tilde{A}+\tilde{B}\cdot \Delta t+\frac{\tilde{D}}{N^2}
\end{equation}

Note that $\Delta t$ and $\Delta u$ are essentially equivalent (up to a constant) near the critical value, so
\begin{equation}
F^{\operatorname{reg}}_N(u)=A+B\Delta u+C(\Delta u)^2+\frac{D}{N^2},
\end{equation}
for some constants $A$, $B$, $C$ and $D$, that come both from the change of variables from $t$ to $u$ and the prefactor in \eqref{FNFN}. 

The function $F^{(0)}(u)$ can actually be written in terms of generalized hypergeometric functions, with the aid of {\sc Maple}:
\begin{equation}
F^{(0)}(u)=6u^2+216u^4\, _4F_3\left(\begin{array}{llll} 1 & 1 & \tfrac{4}{3} & \tfrac{5}{3}\\[1mm] 2 & \tfrac{5}{2} & 3 & \end{array}; 34992u^4\right)
+13608u^6 \, _4F_3\left(\begin{array}{llll} 1 & \tfrac{3}{2} & \tfrac{11}{6} & \tfrac{13}{6}\\[1mm] 3 & \tfrac{5}{2} & \tfrac{7}{2} & \end{array}; 34992u^4\right)
\end{equation}

This representation is not particularly useful for computations, but it allows us to determine the regularity of the function $F^{(0)}(u)$ at the critical point. Using standard arguments, see for instance \cite[\S 16.2]{DLMF}, this function is absolutely convergent when $34992u^4=1$, which corresponds in particular to $u=u_c$. Then we have
\begin{equation}
A=F^{(0)}(u_c)=54\sqrt{3}+\frac{1}{162}\, _4F_3\left(\begin{array}{llll} 1 & 1 & \tfrac{4}{3} & \tfrac{5}{3}\\[1mm] 2 & \tfrac{5}{2} & 3 & \end{array}; 1\right)
+\frac{7\sqrt{3}}{5832} \, _4F_3\left(\begin{array}{llll} 1 & \tfrac{3}{2} & \tfrac{11}{6} & \tfrac{13}{6}\\[1mm] 3 & \tfrac{5}{2} & \tfrac{7}{2} & \end{array}; 1\right).
\end{equation}

Differentiation with respect to $u$ increases each parameter of the hypergeometric functions by one, and it is not difficult to check that $F^{(0)}(u)$ is twice differentiable at $u=u_c$. Therefore, we get
\begin{equation}
 B=F^{(0)'}(u_c), \qquad C=\frac{1}{2} F^{(0)''}(u_c).
\end{equation}

The non--analytic terms in $\tilde{F}_N(t)$ can be determined from the double scaling regime. Consider the Toda equation \eqref{Toda}, which is valid for any $t>t_c$, see \cite{BD}. We can establish a connection between the variables $\lambda$ and $t$, from the double scaling setting and the change of variables \eqref{tu}. Namely, if $\Delta u=u_c-u$ and $\Delta t=t-t_c$, we have
\begin{equation}
 \Delta u= 2^{-7/3} 3^{-7/4}\Delta t+\mathcal{O}((\Delta t)^2)
\end{equation}
and since $\lambda=-2^{12/5}3^{7/4}N^{4/5}\Delta u $, from the double scaling relation, we get
\begin{equation}
\lambda=-2^{1/15}N^{4/5}\Delta t+\mathcal{O}(N^{-4/5}).
\end{equation}

We define the auxiliary variable 
\begin{equation}\label{nut}
\nu=-2^{1/15}N^{4/5}\Delta t, 
\end{equation}
and observe that $\nu-\lambda=\mathcal{O}(N^{-4/5})$. The Toda equation \eqref{Toda} can be written in terms of this variable $\nu$:
\begin{equation}\label{Todanu}
\frac{d^2 \tilde{F}_N(\nu)}{d\nu^2}=2^{-2/15} N^{-8/5}\tilde{\gamma}^2_{N,N}(\nu),
\end{equation}
and then, since
\begin{equation}
\tilde{\gamma}^2_{N,N}(t)=\frac{1}{2\sqrt{t}}\gamma^2_{N,N}(u),
\end{equation}
we have
\begin{equation}
\tilde{\gamma}^2_{N,N}(\nu)=\frac{1}{2\sqrt{t_c}}\gamma^2_{N,N}(\nu)+\mathcal{O}(N^{-4/5})
=2^{-2/3}3^{-1/2}\gamma^2_{N,N}(\nu)+\mathcal{O}(N^{-4/5}).
\end{equation}

Then, substituting the asymptotic expansion for $\gamma^2_{N,N}(\nu)$ in \eqref{Todanu}, we get the following differential equation:
\begin{equation}
\begin{aligned}
\frac{d^2 \tilde{F}_N(\nu)}{d\nu^2}&=\frac{2^{-4/5}3^{-1/2}}{N^{8/5}} \gamma^2_{N,N}(\nu)+\mathcal{O}(N^{-12/5})=\frac{2^{-4/5}}{N^{8/5}}-\frac{y_{\alpha}(\nu)}{N^2}+\mathcal{O}(N^{-12/5}).
\end{aligned}
\end{equation}

Integrating twice in $\nu$, we have
\begin{equation}
\tilde{F}_N(\nu)=\frac{2^{-9/5}\nu^2}{N^{8/5}} +A_N\nu+B_N-\frac{Y_{\alpha}(\nu)}{N^2}+\mathcal{O}(N^{-12/5}).
\end{equation}

Here $A_N$ and $B_N$ are constants of integration. In terms of $\Delta t$, using \eqref{nut}, we have
\begin{equation}\label{FNt_ds}
\begin{aligned}
\tilde{F}_N(t)&=2^{-5/3}(\Delta t)^2- A_N 2^{1/15}N^{4/5}\Delta t+B_N-\frac{Y_{\alpha}(\nu)}{N^2}+\mathcal{O}(N^{-12/5}),\\
&=\tilde{F}^{\operatorname{reg}}_N(t)+\tilde{F}^{\operatorname{sing}}_N(\nu)+\mathcal{O}(N^{-12/5}),
\end{aligned}
\end{equation}
where the term $\tilde{F}^{\operatorname{reg}}_N(t)$ matches the one before, and
\begin{equation}
\tilde{F}^{\operatorname{sing}}_N(\nu)
=-\frac{Y_{\alpha}(\nu)}{N^2}.
\end{equation}

This function $Y_{\alpha}(\nu)$ satisfies the ODE
\begin{equation}\label{Yy}
Y''_{\alpha}(\nu)=y_{\alpha}(\nu), 
\end{equation}
with boundary condition
\begin{equation}\label{Yy_bc}
Y_{\alpha}(\nu)=\frac{2\sqrt{6}}{45}(-\nu)^{5/2}-\frac{1}{48}\ln(-\nu)+\mathcal{O}(\nu^{-5/2}), \qquad (-\nu)\to\infty.
\end{equation}
 
This boundary condition is chosen to match the two non--analytic terms at $t=t_c$ appearing in the regular regime. Finally, we use the fact that $\nu=\lambda+\mathcal{O}(N^{-4/5})$, and define
\begin{equation}
F^{\operatorname{sing}}_N(\lambda)=-\frac{Y_{\alpha}(\lambda)}{N^2},
\end{equation}
which completes the proof of the theorem.

\appendix
\section{Proof of Theorem \ref{Th_Phi}}\label{ApA}
We start with the Riemann--Hilbert problem for the function $\Phi(\z;\lambda,\alpha)$. First of all, we translate the contour $\Gamma_{\Psi}$ to $\Gamma_{\Phi}=\Gamma_{\Psi}+\z_0$, where $\z_0=-2/\sqrt{6}$, recall \eqref{zeta0}. This implies a modification of $\Phi$ in different sectors taking into account the jumps on the contour $\Gamma_{\Psi}$, but it does not alter the asymptotic behavior. Thus, we consider the following RH problem for $\Phi(\z;\lambda,\alpha)$:
\begin{enumerate}
\item $\Phi$ is analytic on $\C\setminus \Ga_{\Phi}$, where the contour $\Gamma_{\Phi}$ is depicted in Figure \ref{PI_jumps}, but centered at the point $\z=\z_0$. 
\item On $\Ga_{\Phi}$, $\Phi$ has the following jumps:
\begin{equation}
\Phi_+(s)=\Phi_-(s)
\begin{cases}
\begin{pmatrix}
1 & 0\\ 1 & 1
\end{pmatrix}, \qquad s\in\gamma_{\pm 2}\\
\begin{pmatrix}
0 & 1\\ -1 & 0
\end{pmatrix}, \qquad s\in\rho\\
\begin{pmatrix}
1 & \alpha \\ 0 & 1
\end{pmatrix}, \qquad s\in\gamma_{1}\\
\begin{pmatrix}
1 & 1-\alpha \\ 0 & 1
\end{pmatrix}, \qquad s\in\gamma_{-1}.
\end{cases}
\end{equation}
\item As $\z\to\infty$, for fixed $\lambda$ and $\alpha$, $\Phi$ expands in the asymptotic series
\begin{equation}
\Phi(\z;\lambda,\alpha)\sim \frac{(\z-\z_0)^{\sg_3/4}}{\sqrt 2}
\begin{pmatrix}
1 & -i \\
1 & i
\end{pmatrix}
\left(I+\sum_{k=1}^\infty \frac{\Phi_k(\lambda,\alpha)}{(-\lambda)^{k/4}\z^{k/2}}\right)e^{(-\lambda)^{5/4}g(\z)\sigma_3},
\end{equation}
where we recall that
\begin{equation}
g(\z)=\frac{4}{5}\left(\z+\frac{2}{\sqrt 6}\right)^{5/2}-\frac{2\sqrt 6}{3}\left(\z+\frac{2}{\sqrt 6}\right)^{3/2},
\end{equation}
with a cut on the axis $(-\infty,\z_0]$.
\end{enumerate}

\subsection{First transformation}
Consider the new matrix function
\begin{equation}
S(\z;\lambda,\alpha)=\Phi(\z;\lambda,\alpha)e^{-(-\lambda)^{5/4}g(\z)\sg_3},
\end{equation}

Then $S(\z;\lambda,\alpha)$ satisfies the following RH problem:

\begin{enumerate}
\item $S$ is analytic on $\C\setminus \Ga_{\Phi}$.
\item On $\Ga_{\Phi}$, $S$ has the following jumps:
\begin{equation}
S_+(s)=S_-(s)
\begin{cases}
\begin{pmatrix}
1 & 0\\ e^{-2(-\lambda)^{5/4}g(s)} & 1
\end{pmatrix}, \qquad s\in\gamma_{\pm 2}\\
\begin{pmatrix}
0 & 1\\ -1 & 0
\end{pmatrix}, \qquad s\in\rho\\
\begin{pmatrix}
1 & \alpha\, e^{2(-\lambda)^{5/4}g(s)}\\ 0 & 1
\end{pmatrix}, \qquad s\in\gamma_{1}\\
\begin{pmatrix}
1 & (1-\alpha) e^{2(-\lambda)^{5/4}g(s)}\\ 0 & 1
\end{pmatrix}, \qquad s\in\gamma_{-1}.
\end{cases}
\end{equation}

\item As $\z\to\infty$, $S$ expands in the asymptotic series
\begin{equation}\label{rhpT2}
S(\z;\lambda,\alpha)\sim \frac{(\z-\z_0)^{\sg_3/4}}{\sqrt 2}
\begin{pmatrix}
1 & -i \\
1 & i
\end{pmatrix}
\left(I+\sum_{k=1}^\infty \frac{\Phi_k(\lambda,\alpha)}{(-\lambda)^{k/4}\z^{k/2}}\right)
\end{equation}
\end{enumerate}

Note that the jump on $\rho$ is the same as before, since the function $g(\z)$ has a jump on $(-\infty,\z_0]$, and $g_+(s)=-g_-(s)$ for $s\in(-\infty,\z_0]$. 

Let us deform the contour slightly near the point $\z=\z_0$, as indicated in Figure \ref{T_PI_jumps_mod}. All the jumps remain the same, except a new one on $\gamma_0$, which is equal to
\begin{equation}
S_+(s)=S_-(s)\begin{pmatrix} 1 & 1\\ 0 &1\end{pmatrix}, \qquad s\in\gamma_0.
\end{equation}

Taking into account the level curves corresponding to $\textrm{Re}\, g(\z)=0$, see Figure \ref{T_PI_jumps_mod}, it is clear that the jumps on $\gamma_{\pm 1}$ and on $\gamma_{\pm 2}$ tend to the identity exponentially fast, provided that the segment $\gamma_0$ does not enter the region where $\textrm{Re}\,g(\z)<0$.

\subsection{Model RH problem}
Now we ignore all the jumps exponentially close to identity (that is, the ones on $\gamma_{\pm 1}$ and on $\gamma_{\pm  2}$), and we solve the model Riemann--Hilbert problem: we look for $M(\z)$ such that
\begin{enumerate}
\item $M$ is analytic on $\C\setminus(-\infty,\z_0]$.
\item On $(-\infty,\z_0)$, we have
\begin{equation}
M_+(s)=M_-(s)\begin{pmatrix} 0 & 1\\ -1 & 0\end{pmatrix}
\end{equation}
\item As $\z\to\infty$, $\z\notin(-\infty,\z_0]$, we have 
\begin{equation}
M(\z)=\frac{(\z-\z_0)^{\sg_3/4}}{\sqrt 2}
\begin{pmatrix}
1 & -i \\
1 & i
\end{pmatrix}\left(I+\mathcal{O}\left(\z^{-1}\right)\right).
\end{equation}
\end{enumerate}

This model RH problem can be solved explicitly, namely
\begin{equation}
M(\z)=\frac{(\z-\z_0)^{\sg_3/4}}{\sqrt{2}}\begin{pmatrix}1& -i\\ 1 & i\end{pmatrix}.
\end{equation}


\begin{figure}
\centerline{
\includegraphics[width=70mm,height=70mm]{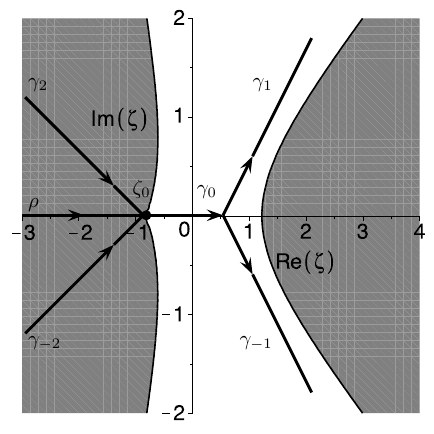}}
 \caption{Level curves where $\textrm{Re}\,g(\z)=0$. In grey, the region where $\textrm{Re}\,g(\z)>0$, in white where $\textrm{Re}\, g(\z)<0$. The modified contour $\Gamma_S=\gamma_{\pm 1}\cup\gamma_{\pm 2}\cup\gamma_0\cup\rho$ is also depicted.}
\label{T_PI_jumps_mod}
\end{figure}

\subsection{Local parametrix around $\z=\z_0$}

Take a fixed disc $D(\z_0,\varepsilon)$, and consider the following RH problem for a function $P(\z;\lambda,\alpha)$:
\begin{enumerate}
\item $P$ is analytic on $D(\z_0,\varepsilon)\cap \Ga_S$.
\item In $D(\z_0,\varepsilon)\cap \Ga_S$, $P$ has the following jumps:
\begin{equation}
P_+(s)=P_-(s)
\begin{cases}
\begin{pmatrix}
1 & 0\\ e^{-2(-\lambda)^{5/4}g(s)} & 1
\end{pmatrix}, \qquad s\in\gamma_{\pm 2}\cap D(\z_0,\varepsilon)\\
\begin{pmatrix}
0 & 1\\ -1 & 0
\end{pmatrix}, \qquad s\in\rho\cap D(\z_0,\varepsilon)\\
\begin{pmatrix} 1 & 1\\ 0 &1\end{pmatrix}, \qquad s\in\gamma_0\cap D(\z_0,\varepsilon)\\
\end{cases}
\end{equation}
\item Uniformly for $\z\in\partial D(\z_0,\varepsilon)$, we have the matching
\begin{equation}\label{matchingPS}
P(\z;\lambda,\alpha)=M(\z)\left(I+\mathcal{O}((-\lambda)^{5/4})\right).
\end{equation}
\end{enumerate}

The reduction to constant jumps is standard now: consider
\begin{equation}
\tilde{P}(\z;\lambda,\alpha)=P(\z;\lambda,\alpha)e^{(-\lambda)^{5/4}g(\z)\sg_3},
\end{equation}
then $\tilde{P}(\z;\lambda,\alpha)$ satisfies the following RH problem:
\begin{enumerate}
\item $\tilde{P}$ is analytic on $D(\z_0,\varepsilon)\cap \Ga_S$.
\item $D(\z_0,\varepsilon)\cap \Ga_S$, $\tilde{P}$ has the following jumps:
\begin{equation}
\tilde{P}_+(s)=\tilde{P}_-(s)
\begin{cases}
\begin{pmatrix}
1 & 0\\ 1 & 1
\end{pmatrix}, \qquad s\in\gamma_{\pm 2}\cap D(\z_0,\varepsilon)\\
\begin{pmatrix}
0 & 1\\ -1 & 0
\end{pmatrix}, \qquad s\in\rho\cap D(\z_0,\varepsilon)\\
\begin{pmatrix} 1 & 1\\ 0 &1\end{pmatrix}, \qquad s\in\gamma_0\cap D(\z_0,\varepsilon)\\
\end{cases}
\end{equation}
\item Uniformly for $\z\in\partial D(\z_0,\varepsilon)$, we have the matching
\begin{equation}\label{matchingPtS}
\tilde{P}(\z;\lambda,\alpha)=M(\z)\left(I+\mathcal{O}((-\lambda)^{-5/4})\right)e^{(-\lambda)^{5/4}g(\z)\sigma_3}.
\end{equation}
\end{enumerate}

This is a standard Airy Riemann--Hilbert problem, so we look for a local parametrix in the form
\begin{equation}
\tilde{P}(\z;\lambda,\alpha)=E(\z;\lambda)A((-\lambda)^{5/6}f(\z)),
\end{equation}
where $E(\z;\lambda)$ is an analytic prefactor, the matrix $A(w)$ is built of suitably chosen Airy functions, and $f(\z)$ is a conformal mapping from a neighborhood of $\z_0$ onto a neighborhood of $0$ in the auxiliary $w$ complex plane. The important fact for us in this case is that $A((-\lambda)^{5/6}f(\z))$ can be expanded asymptotically as the argument grows large:
\begin{equation}
A((-\lambda)^{5/6}f(\z))=((-\lambda)^{5/6}f(\z))^{-\sg_3/4}\frac{1}{2\sqrt{\pi}}\begin{pmatrix} 1 & i\\ i & 1\end{pmatrix}
\left(I+\mathcal{O}((-\lambda)^{-5/4})\right)e^{-\frac{2}{3}(-\lambda)^{5/4}(f(\z))^{3/2}\sigma_3}.
\end{equation}

Now we choose $f(\z)$ in the following way:
\begin{equation}
f(\z)=\left(-\frac{3}{4}g(\z)\right)^{3/2},
\end{equation}
to match the exponential factors, and since from \eqref{rhp18} we can write
\begin{equation}
g(\z)=-\frac{2\sqrt{6}}{3}(\z-\z_0)^{3/2}\left(1-\frac{\sqrt{6}}{5}(\z-\z_0)\right),
\end{equation}
then it follows that $f(\z)$ is indeed a conformal mapping in a neighborhood of $\z=\z_0$. Finally, to get the correct matching with $M(\z)$ in \eqref{matchingPtS}, we take
\begin{equation}
E(\z;\lambda)=\sqrt{\pi}\,M(\z)\begin{pmatrix}1 & -i\\ -i & 1\end{pmatrix}((-\lambda)^{5/6}f(\z))^{-\sigma_3/4},
\end{equation}
which is analytic for $\z\in(-\infty,\z_0]$ because of the jumps of $M(\z)$ and the fractional power $f(\z)^{-\sigma_3/4}$.
\subsection{Second transformation}
Finally, we define
\begin{equation}
R(\z;\lambda,\alpha)=S(\z;\lambda,\alpha)
\begin{cases}
M^{-1}(\z),& \qquad \z\in\mathbb{C}\setminus D(\z_0,\varepsilon),\\
P^{-1}(\z),& \qquad \z\in D(\z_0,\varepsilon).
\end{cases}
\end{equation}

Then this matrix $R(\z;\lambda,\alpha)$ is analytic in $\mathbb{C}\setminus \Gamma_R$, where $\Gamma_R$ is depicted in Figure \ref{R_PI_jumps}. It has exponentially small jumps in $(-\lambda)$ in the whole contour except on the boundary of the disc, $\partial D(\z_0,\varepsilon)$ where it is of order $(-\lambda)^{5/4}$.

\begin{figure}
\centerline{
\includegraphics{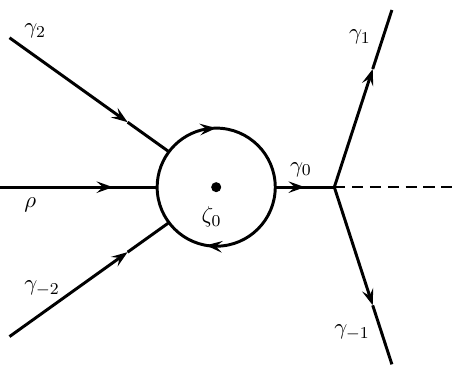}}
 \caption{The final contour $\Gamma_R$.}
\label{R_PI_jumps}
\end{figure}

Following standard arguments, see for instance \cite{Ble}, we can conclude that as $\lambda\to-\infty$, we have
\begin{equation}
R(\z;\lambda,\alpha)
=I+\mathcal{O}\left(\frac{1}{(-\lambda)^{5/4}(1+|\z|)}\right),
\end{equation}
uniformly for $\z\in\mathbb{C}\setminus D(\z_0,\varepsilon)$. Now, undoing the transformations, we get
\begin{equation}
\begin{aligned}
\Phi(\z;\lambda,\alpha)&=S(\z;\lambda,\alpha)e^{(-\lambda)^{5/4}g(\z)\sg_3}\\
&=R(\z;\lambda,\alpha)M(\z)e^{(-\lambda)^{5/4}g(\z)\sg_3}\\
&=\left(I+\mathcal{O}\left(\frac{1}{(-\lambda)^{5/4}(1+|\z|)}\right)\right)M(\z)e^{(-\lambda)^{5/4}g(\z)\sg_3},
\end{aligned}
\end{equation}
valid as $(-\lambda)\to\infty$, uniformly for $\z\in\mathbb{C}\setminus D(\z_0,\varepsilon)$. This completes the proof of Theorem \ref{Th_Phi}.

\end{document}